\PassOptionsToPackage{table,prologue}{xcolor}
\documentclass[sigconf]{acmart}
\AtBeginDocument{%
  }

\usepackage{multirow}
\usepackage{makecell}
\usepackage{pifont}
\usepackage{bm}
\usepackage{algorithm}
\usepackage{algpseudocode}
\usepackage{algorithmicx}
\usepackage{circledsteps}
\pgfkeys{/csteps/fill color=gray!30}

\algrenewcommand\algorithmicrequire{\textbf{Input:}}
\algrenewcommand\algorithmicensure{\textbf{Output:}}

\graphicspath{ {./figures/} }

\begin{document}

\title{OFL: Opportunistic Federated Learning for Resource-Heterogeneous and Privacy-Aware Devices}

\author{Yunlong Mao, Mingyang Niu, Ziqin Dang, Chengxi Li, Hanning Xia, Yuejuan Zhu, Haoyu Bian, Yuan Zhang, Jingyu Hua, Sheng Zhong}
\affiliation{%
  \institution{Nanjing University}
  \city{Nanjing}
  \country{China}
}

\renewcommand{\shortauthors}{Yunlong Mao, Mingyang Niu, Ziqin Dang et al.}

\begin{abstract}
Efficient and secure federated learning (FL) is a critical challenge for resource-limited devices, especially mobile devices. Existing secure FL solutions commonly incur significant overhead, leading to a contradiction between efficiency and security. As a result, these two concerns are typically addressed separately. This paper proposes Opportunistic Federated Learning (OFL), a novel FL framework designed explicitly for resource-heterogenous and privacy-aware FL devices, solving efficiency and security problems jointly. OFL optimizes resource utilization and adaptability across diverse devices by adopting a novel hierarchical and asynchronous aggregation strategy. OFL provides strong security by introducing a differentially private and opportunistic model updating mechanism for intra-cluster model aggregation and an advanced threshold homomorphic encryption scheme for inter-cluster aggregation. Moreover, OFL secures global model aggregation by implementing poisoning attack detection using frequency analysis while keeping models encrypted. We have implemented OFL in a real-world testbed and evaluated OFL comprehensively. The evaluation results demonstrate that OFL achieves satisfying model performance and improves efficiency and security, outperforming existing solutions.
\end{abstract}





\maketitle

\section{Introduction}
Federated Learning (FL) has emerged as a promising solution for collaborative machine learning, especially for large-scale participants. For example, Google uses FL for improving autocorrect and text suggestions of mobile phone keyboards \cite{xu2023federated}, and Nvidia offers FLARE \cite{NVIDIA-2025-03-08} to enable FL in autonomous vehicles driving. However, FL has encountered critical issues in practical deployments \cite{chahoud2023feasibility, ji2023joint, ezzeldin2023fairfed}, including device heterogeneity, resource limitation, data heterogeneity, and so on. What makes things worse is that private data leakage has been identified in FL. Although raw data is kept locally, private training data can still be disclosed by membership inference \cite{yan2022membership}, reconstruction \cite{mehnaz2022your}, and many other attacks.

To implement efficient mobile FL systems, clustering FL (CFL) \cite{sattler2020clustered} and hierarchical FL (HFL) \cite{wang2021resource} frameworks have been proposed. CFL frameworks like IFCA \cite{ghosh2020efficient} suggest clustering participants based on their training data similarity to mitigate the influence of data heterogeneity. Since real-world participants' computing and communicating resources may vary dramatically (e.g., a gap of 160x in Asteroid \cite{ye2024asteroid}), frameworks like \cite{xu2023clustered} suggest grouping participants based on device resources. HFL frameworks commonly group participants based on the similarity of local model updates \cite{briggs2020federated} or datasets characteristics \cite{wang2021resource} while dynamically adjusting groups in response to changes in network conditions or client availability \cite{abad2020hierarchical, chen2020fedcluster}. Except for differences in grouping, HFL assigns a stable leader for each group to communicate with the central server, while CFL has no restrictions on communications.

Meanwhile, research efforts \cite{sav2020poseidon, yanefficient, kanagavelu2020two, fereidooni2021safelearn} have been made to address privacy leakage issues, introducing sophisticated defense solutions. Since cryptographic tools like multiparty computation (MPC) \cite{gehlhar2023safefl} and homomorphic encryption (HE) \cite{jin2023fedml} are widely used, these solutions need intensive computation and communication, leading to unaffordable costs for FL participants, especially for mobile devices. Besides, it has been proved that solely applying cryptographic tools cannot prevent threats from inference and reconstruction attacks \cite{zhao2024loki,mehnaz2022your}. Thus, artificial perturbations like differentially private (DP) mechanisms are still needed to protect models trained in FL \cite{truex2020ldp, zhao2020local,mao2023secure,yang2023privatefl}.

Although CFL and HFL consider participant conditions, they rely on frequent communication between participants and the central server. This leads to significant communication overhead and stable communication channel requirements, which is impractical for mobile devices like autonomous vehicles and mobile phones. Besides, the existing FL frameworks lack a comprehensive consideration of device heterogeneity and security. Most existing studies either focus on matching heterogeneous resources or defending against attacks, failing to provide heterogeneous device management and FL security in a single solution. However, we note that it is rather challenging to design a satisfying FL framework due to the following reasons. \Circled{1} Heterogeneous devices have a large variance of local running time. It is hard to synchronize on-device models efficiently. \Circled{2} Participants, especially for mobile devices, may have unstable network connections or limited bandwidth, leading to unavailability or synchronizing dropout. \Circled{3} Preserving local data privacy contradicts global model learning algorithms. It is hard to balance data privacy and model performance. \Circled{4} Securing data confidentiality and integrity against an untrusted server contradicts malicious client detection. Encryption-based FL solutions have trouble finding malicious model updates by leveraging conventional methods efficiently.

Having realized the gap between existing FL frameworks and practical demands, we propose an opportunistic federated learning (OFL) framework for resource-heterogeneous and privacy-aware devices, solving the aforementioned challenges simultaneously. OFL uses a hierarchical aggregation protocol to enhance the performance of FL in the presence of heterogeneous devices and data distributions. Notably, multiple factors, including computation, communication, data heterogeneity, availability, and privacy budget, are fully considered by OFL to fulfill practical demands. OFL ensures efficiency by employing an opportunistic updating strategy. When a device is available to communicate, it seizes the opportunity to sync the most crucial part of model parameters based on a well-defined score function. This strategy, coupled with the pruning of updating gradients according to significance and similarity, ensures that the server in OFL can maintain a balanced model performance through periodic re-clustering based on local model states and device availability. For security concerns, OFL utilizes a hybrid solution, combining an advanced threshold fully HE scheme for inter-cluster aggregation with a novel DP mechanism for intra-cluster aggregation. In particular, we design the DP mechanism utilizing the opportunistic updating strategy and pruning, which offers a better balance between privacy budget and model performance than conventional DP mechanisms. Overall, OFL makes the following contributions:
\begin{itemize}
    \item OFL provides a unified solution for resource-heterogeneous and privacy-aware devices, fully considering each device's resource limitation and data heterogeneity.
    \item OFL adopts an opportunistic model updating strategy, providing stragglers with an optimized syncing option considering device resources and privacy budgets, achieving better resource utilization.
    \item OFL has a hybrid defense with provable security, including a novel DP mechanism and an advanced threshold fully HE scheme. In addition, OFL significantly reduces the additional overhead of protection through the opportunistic updating strategy.
    \item We implement OFL in a real-world testbed using Jetson Xavier NX boards mimicking different devices with various working modes. We evaluate OFL on various learning tasks and settings. Real-world test results and large-scale simulating results indicate that OFL provides satisfying features and outperforms existing work. The source code of OFL implementation will be available after anonymous peer review.
\end{itemize}

\section{Related Work and Motivation}
\begin{table*}[]
\caption{Comparison of recent solutions for resource-efficient FL.}
\label{tab:related}
\footnotesize
\begin{tabular}{cccccccccccc}
\toprule
\multicolumn{1}{c}{\multirow{2}{*}{Solution}} & \multicolumn{5}{c}{Resource Perception} & \multicolumn{1}{c}{\multirow{2}{*}{H}} & \multicolumn{1}{c}{\multirow{2}{*}{S}} & \multicolumn{1}{c}{\multirow{2}{*}{AS}} & \multicolumn{1}{c}{\multirow{2}{*}{CS}} & \multicolumn{1}{c}{\multirow{2}{*}{Computation}} & \multicolumn{1}{c}{\multirow{2}{*}{Communication}} \\ \cmidrule(lr){2-6}
\multicolumn{1}{c}{} & C & N & D & A & P & \multicolumn{1}{c}{} & \multicolumn{1}{c}{} & \multicolumn{1}{c}{} & \multicolumn{1}{c}{} & \multicolumn{1}{c}{} & \multicolumn{1}{c}{} \\ \midrule

CFL \cite{sattler2020clustered} & \ding{55} & \ding{55} & \ding{51} & \ding{55} & \ding{55} & \ding{55} & \ding{55} & \ding{55} & \ding{55} 
& \makecell{w: $\mathcal{O}(|\bm{\theta}|)$ \\ s: $\mathcal{O}( (N_k^2 + N\log_{2}N) \cdot |\bm{\theta}|)$} 
& \makecell{w: $\mathcal{O}(|\bm{\theta}|)$ \\ s: $\mathcal{O}(N \cdot |\bm{\theta}|)$} 
\\ \cmidrule(lr){1-1} \cmidrule(lr){2-10} \cmidrule(lr){11-12}

IFCA \cite{ghosh2020efficient} & \ding{55} & \ding{55} & \ding{51} & \ding{55} & \ding{55} & \ding{55} & \ding{55} & \ding{51} & \ding{55} 
& \makecell{w: $\mathcal{O}(T_{lt} \cdot |\bm{\theta}|)$  \\ s: $\mathcal{O}(K \cdot N+ K \cdot T_{lt} \cdot |\bm{\theta}|)$ } 
& \makecell{w: $\mathcal{O}(T_{lt} \cdot |\bm{\theta}| )$
 \\ s:$\mathcal{O}( (T_{lt} + N) \cdot |\bm{\theta}| )$} 
\\ \cmidrule(lr){1-1} \cmidrule(lr){2-10} \cmidrule(lr){11-12}


FLDP \cite{stevens2022efficient} & \ding{51} & \ding{55} & \ding{55} & \ding{51} & \ding{51} & \ding{55} & \ding{51} & \ding{55} & \ding{51} 
& \makecell{w: $\mathcal{O}( N\log_{2}N + |\bm{\theta}|^{2})$ \\ s: $\mathcal{O}(N \cdot (\log_{2}N + |\bm{\theta}|) + |\bm{\theta}|^{2})$} 
& \makecell{w: $\mathcal{O}(N + |\bm{\theta}|)$ \\ s:$\mathcal{O}((N \cdot |\bm{\theta}|+|\bm{\theta}|)$} 
\\ \cmidrule(lr){1-1} \cmidrule(lr){2-10} \cmidrule(lr){11-12}


FedPHE \cite{yanefficient} & \ding{55} & \ding{55} & \ding{51} & \ding{51} & \ding{51} & \ding{55} & \ding{51} & \ding{51} & \ding{55} 
& \makecell{w: $ \mathcal{O}(T_{lt} \cdot |\bm{\theta}| + |\bm{\theta}|^{2} )$ \\ s: $ \mathcal{O}((N+N_k) \cdot |\bm{\theta}| + |\bm{\theta}|^{2} )$} 
& \makecell{w: $ \mathcal{O}( |\bm{\theta}| + |\bm{\theta}|^{2})$ \\ s: $\mathcal{O}(N \cdot |\bm{\theta}| + N \cdot |\bm{\theta}|^{2})$ } 
\\ \cmidrule(lr){1-1} \cmidrule(lr){2-10} \cmidrule(lr){11-12}

Asteroid \cite{ye2024asteroid} & \ding{51} & \ding{51} & \ding{51} & \ding{55} & \ding{55} & \ding{55} & \ding{55} & \ding{55} & \ding{55} 
& \makecell{w: $ \mathcal{O}( T_{tl} \cdot |\bm{\theta}| )$ \\ s: $ \mathcal{O}( K^3 \cdot |\bm{\theta}| +  N \cdot |\bm{\theta}|)$} 
& \makecell{w: $ \mathcal{O}(T_{tl} \cdot |\bm{\theta}|)$ \\ s: $\mathcal{O}( N \cdot |\bm{\theta}|)$} 
\\ \cmidrule(lr){1-1} \cmidrule(lr){2-10} \cmidrule(lr){11-12}

RFL-HA \cite{wang2021resource} & \ding{51} & \ding{51} & \ding{51} & \ding{55} & \ding{55} & \ding{51} & \ding{55} & \ding{55} & \ding{55} 
& \makecell{w: $\mathcal{O}(|\bm{\theta}|)$ \\ l: $\mathcal{O}(N_k \cdot |\bm{\theta}|)$ \\ s: $\mathcal{O}(K \cdot |\bm{\theta}| + N \cdot K )$} 
& \makecell{w: $\mathcal{O}(|\bm{\theta}|)$ \\ l: $\mathcal{O}(N_k \cdot |\bm{\theta}|)$ \\ s: $\mathcal{O}(K \cdot |\bm{\theta}| )$} 
\\ \cmidrule(lr){1-1} \cmidrule(lr){2-10} \cmidrule(lr){11-12}

AHFL \cite{guo2023privacy} & \ding{51} & \ding{51} & \ding{55} & \ding{55} & \ding{51} & \ding{51} & \ding{55} & \ding{55} & \ding{55} 
& \makecell{w:  $\mathcal{O}(|\bm{\theta}|)$ \\ l:  $\mathcal{O}(N_k \cdot |\bm{\theta}|)$\\ s: $\mathcal{O}(K \cdot N + K \cdot |\bm{\theta}|)$} 
& \makecell{w: $\mathcal{O}(|\bm{\theta}|)$ \\ l: $\mathcal{O}(N_k \cdot |\bm{\theta}|)$ \\ s: $\mathcal{O}(K \cdot |\bm{\theta}|)$} 
\\ \cmidrule(lr){1-1} \cmidrule(lr){2-10} \cmidrule(lr){11-12}


\rowcolor{gray!30} \textbf{OFL} & \ding{51} & \ding{51} & \ding{51} & \ding{51} & \ding{51} & \ding{51} & \ding{51} & \ding{51} & \ding{51} 
& \makecell{w: $\mathcal{O}(T_{lt} \cdot |\bm{\theta}|)$ \\ l: $\mathcal{O}(N_k \cdot |\bm{\theta}|)$ \\ s: $\mathcal{O}(N \cdot K + K \cdot |\bm{\theta}| \cdot \log_{2}|\bm{\theta}| )$} 
& \makecell{w: $\mathcal{O}((\gamma_{up} + \gamma_{down}) \cdot |\bm{\theta}|)$ \\ l: $\mathcal{O}((\gamma_{up} + \gamma_{down}) \cdot N_k \cdot |\bm{\theta}|+|\bm{\theta}|)$ \\ s: $\mathcal{O}(K \cdot |\bm{\theta}| )$}
 \\ \bottomrule
\multicolumn{12}{l}{ $\bullet$ ``C'' stands for computing resource, ``N'' for network, ``D'' for data distribution, ``A'' for availability, ``P'' for privacy budget.} \\
\multicolumn{12}{l}{ $\bullet$ ``H'' indicates hierarchical FL, ``S'' indicates whether the solution provides strong security guarantees.} \\
\multicolumn{12}{l}{$\bullet$ ``AS'' stands for asynchronous aggregation support, ``CS'' for the resistance to collusion between the server and clients.} \\
\multicolumn{12}{l}{$\bullet$ ``w'' represents worker, ``s'' represents the server, and ``l'' represents cluster leader.} \\
\multicolumn{12}{l}{$\bullet$ $N$ is The total number of clients, $K$ is the number of clusters, $N_k$ is the number of clients in each cluster.} \\
\multicolumn{12}{l}{$\bullet$ $T_{lt}$ is the local training iterations, $|\bm{\theta}|$ is the size of the model parameters.} \\
\multicolumn{12}{l}{$\bullet$ $\gamma_{up},\gamma_{down}$ are uploading and downloading ratios to the whole model parameters, respectively.} \\
\end{tabular}
\end{table*}

After the emergence of FL, data heterogeneity issues have been identified. Heterogeneous data may lead to model staleness, even resulting in non-convergence. Plenty of approaches have been proposed to mitigate this problem. A typical CFL \cite{sattler2020clustered} groups participants into clusters with jointly trainable data distributions. Another typical work \cite{ghosh2020efficient} proposes an iterative federated clustering algorithm (IFCA), which clusters users and optimizes model parameters to aggregate data within the same cluster. Beyond data, resource heterogeneity has also become an underlying issue of FL, resulting in resource-limited clients straggling. Recent solutions such as Asteroid \cite{ye2024asteroid}, RFL-HA \cite{wang2021resource}, and AHFL \cite{guo2023privacy} have considered clients' resource states when grouping them. By solving an optimizing problem of resource cluster and training job assignment, the improved FL frameworks can achieve better efficiency.

On the other hand, data privacy and aggregation security have also attracted much attention. Recent work like FLDP \cite{stevens2022efficient} has perfectly integrated DP into FL, providing provable data privacy for FL clients. However, aggregation security is not considered in \cite{stevens2022efficient}. FedPHE \cite{yanefficient} provides an impressive solution, ensuring data privacy and aggregation security together. However, the cost is relatively high when compared to other efficient FL solutions. It should be noted that RFL-HA \cite{wang2021resource} and AHFL \cite{guo2023privacy} have applied a hierarchical aggregation scheme for FL, which significantly mitigates the workload of workers and the server.

Through a comprehensive literature review, we find there is a gap between resource-efficient FL and secure FL. It is intuitive to understand. Equipping FL with data privacy preservation and aggregation security means the adoption of heavy cryptographic tools like FHE. We summarize typical papers of related work from multiple perspectives\footnote{Please note that the approximate complexities are calculated mainly in client scale and model size. Therefore, some operations, like accessing an array for constant times and HE algorithms, are not precisely measured.} in Table~\ref{tab:related}. As we can see, CFL solutions like CFL \cite{sattler2020clustered} and IFCA \cite{ghosh2020efficient} mainly focus on solving data heterogeneity issues through client clustering. Secure FL solutions like FLDP \cite{stevens2022efficient} and FedPHE commonly utilize cryptographic tools like FHE and DP, providing security but introducing high costs. Meanwhile, these studies have not fully considered resource limitations and server-client collision situations. Asteroid \cite{ye2024asteroid} has done impressive work on resource optimization and job assignment. However, no security guarantees are provided. Motivated by the practical needs of an efficient FL framework for resource-limited and privacy-aware clients, we design OFL, providing resource perception, privacy preservation, aggregation security, efficiency, and scalability all in one solution.

\section{Preliminary}
\subsection{Federated Learning}
FL aims at a global model collaboratively trained by multiple contributors. In a typical FL \cite{mcmahan2017communication} framework with $N$ clients, each client independently trains its local model using private dataset $\mathcal{D}$. The global optimization objective of FL is
\begin{align}
\label{fl_optim}
    \min_{\bm{\theta}_1, \bm{\theta}_2, \ldots, \bm{\theta}_N} \frac{1}{N} \sum\nolimits_{i=1}^N  \mathcal{L}_i(\bm{\theta}_i;\mathcal{D}_i),
\end{align}
where $\mathcal{L}_i$ is the loss function and $\bm{\theta}_i$ is the local model parameters of the client with index $i$. A central parameter server (PS) aggregates the updates of client local models using a predefined strategy like \textit{FedAvg}, which runs the aggregation based on each client's data size,
\begin{align}
\label{fedavg}
    \overline{\bm{\theta}} \leftarrow \sum\nolimits_{i=1}^N \frac{|\mathcal{D}_i|}{|\overline{\mathcal{D}}|} \bm{\theta}_i,
\end{align}
where $\overline{\mathcal{D}}$ is the union set size of all training data, defined as $|\overline{\mathcal{D}}| = \sum_{i=1}^N |\mathcal{D}_i|$. However, this aggregation strategy may have difficulty in data heterogeneity. Recent research efforts suggest tackling this issue by implementing clustered FL or personalized FL techniques to acquire better performance.

\subsection{Differential Privacy}
DP is a rigorous notation providing strong privacy guarantees \cite{dwork2014algorithmic}. The key idea is to ensure that the output of a computation does not significantly differ when any single individual's data is included or excluded. A randomized algorithm $\mathcal{A}$ is $(\epsilon, \delta)$-DP if for all datasets $\mathcal{D}_1$ and $\mathcal{D}_2$ differing on at most one data point, and for all subsets of outputs $S \subseteq \textrm{Range}(\mathcal{A})$,
\begin{equation}
\Pr[\mathcal{A}(\mathcal{D}_1) \in S] \leq e^\epsilon \cdot \Pr[\mathcal{A}(\mathcal{D}_2) \in S] + \delta,
\end{equation}
where $\epsilon$ is the privacy budget parameter that quantifies the privacy guarantee, with smaller values indicating stronger privacy. The parameter $\delta$ is a small probability that allows for a relaxation of the strict privacy guarantee.

\subsection{Fully Homomorphic Encryption}
FHE is an encryption scheme that allows computations to be carried out on ciphertexts. After the evaluation of the ciphertext, FHE generates an encrypted result that, when decrypted, matches the result of evaluation performed on the plaintext, e.g., $\textsf{Enc}(a) \oplus \textsf{Enc}(b) = \textsf{Enc}(a \oplus b)$. Among FHE schemes, CKKS \cite{cheon2017homomorphic} is a special one, providing arithmetic evaluations of approximate numbers (a.k.a., AHE). In our work, an advanced threshold FHE (ThFHE) based on CKKS is utilized. Here, we recall the ThFHE definition in \cite{passelegue2025low}.
\begin{definition}{ThFHE} The ThFHE scheme is a tuple of six PPT algorithms, i.e., \textsf{ThFHE = (KeyGen, Enc, Eval, ServerDec, PartDec, FinDec)}, satisfying the following properties.
\begin{itemize}
    \item \textsf{KeyGen}($\lambda,N$) takes as input a security parameter $\lambda$, the party number $N$, and returns a public parameter set $\textsf{pp}$ including the plaintext space $\mathcal{M}$, the ciphertext space $\mathcal{C}$, the ciphertext space yielded by \textsf{ServerDec} process $\mathcal{D}_{dec}$, public key $\textsf{pk}$, evaluation key $\textsf{ek}$, and secret key shares $\textsf{sk}_1, \textsf{sk}_2, \ldots, \textsf{sk}_N$.
    \item \textsf{Enc}(\textsf{pp},\textsf{pk},m) takes as input the public parameter set \textsf{pp}, public key \textsf{pk}, plaintext $m \in \mathcal{M}$ and returns a ciphertext $\textsf{ct} \in \mathcal{C}$.
    \item \textsf{Eval}(\textsf{pp},\textsf{ek},\textsf{fun},$\textsf{ct}_1,\textsf{ct}_2,\ldots,\textsf{ct}_n$) takes as input the public parameter set \textsf{pp}, evaluation key \textsf{ek}, function \textsf{fun}: $\mathcal{M}^n \rightarrow \mathcal{M}$ with $n \geq 0$, ciphertexts $\textsf{ct}_1,\textsf{ct}_2,\ldots,\textsf{ct}_n$ and returns a ciphertext $\textsf{ct} \in \mathcal{C}$.
    \item \textsf{ServerDec}(\textsf{pp},\textsf{pk},\textsf{ct}) takes as input the public parameter set \textsf{pp}, public key \textsf{pk}, ciphertext \textsf{ct} and returns a ciphertext $\textsf{ct}_{dec} \in \mathcal{C}_{dec}$.
    \item \textsf{PartDec}(\textsf{pp},$\textsf{sk}_i$,$\textsf{ct}_{dec}$) takes as input the public parameter set \textsf{pp}, secret key share $\textsf{sk}_i$, ciphertext $\textsf{ct}_{dec}$ and returns a partial decryption $\textsf{pd}_{i} \in \mathcal{M}_{share}$, where $\mathcal{M}_{share}$ is the space of partial decryption shares.
    \item \textsf{FinDec}(\textsf{pp},$\textsf{pd}_1,\textsf{pd}_2,\ldots,\textsf{pd}_N$) takes as input the public parameter set \textsf{pp}, partial decryption shares $\textsf{pd}_{i} \in \mathcal{M}_{share}$, $i \in [1,N]$ and returns a plaintext $m^{\prime} \in \mathcal{M} \cup \{\bot \}$.
\end{itemize}
\end{definition}
For a brief, we will omit \textsf{pp} in most cases and abbreviate \textsf{Eval} of \textsf{fun} to \textsf{Fun}, such as \textsf{Add}($\textsf{ct}_1,\textsf{ct}_2,\ldots,\textsf{ct}_n$).

\subsection{System Model and Threat Model}
Unlike the original FL, OFL employs a hierarchical aggregation protocol. This protocol clusters clients based on their resource status and selects a relatively high-performance client as the cluster leader. After local training, clients within a cluster transmit their updated local models and current resource status to the leader in an asynchronous way. Subsequently, the leader updates the leader model on receiving a local model every time. This asynchronous aggregation in OFL is called intra-cluster aggregation. The PS initials a global synchronous aggregation periodically. When the synchronous aggregation begins, each cluster leader transmits the latest leader model and clients' resource status to the PS. The PS will decide whether to re-cluster all clients based on the leader model states and resource reports collected. Meanwhile, the PS performs global model aggregation and sends the updated global model back to all leaders, who then forward it to clients inside clusters. This globally synchronous aggregation in OFL is called inter-cluster aggregation. We demo this hierarchical aggregation in Figure~\ref{fig:system}.
\begin{figure}[t!]
    \centering
    \includegraphics[width=\linewidth]{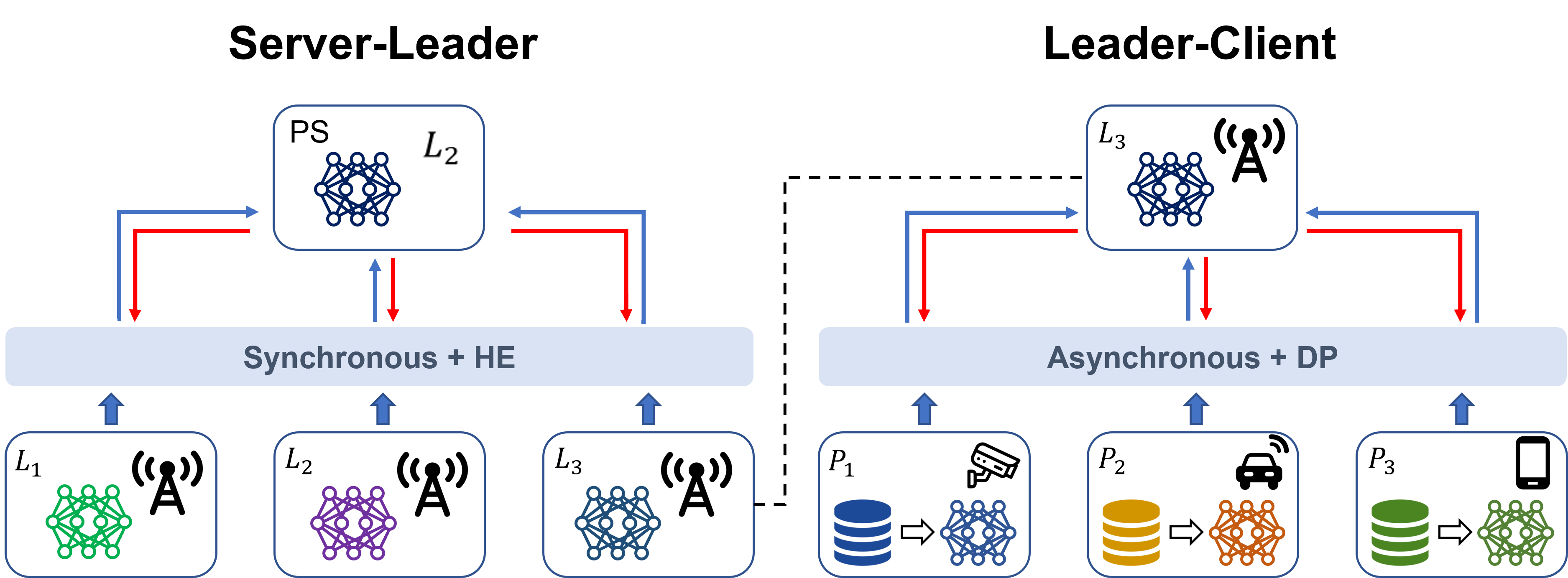}
    \caption{The system model of OFL.}
    \label{fig:system}
\end{figure}

We use the semi-honest threat model in OFL, where the PS, leaders, and clients can be curious about some specific client's private data but must follow the designed protocol honestly. To disclose the private data of the target, the adversary is capable of computing in addition to the protocol, such as membership inference \cite{yan2022membership} and data reconstruction attacks \cite{zhao2024loki}. However, the computation of the protocol cannot be forged arbitrarily. Under this threat model, we also take into account poisoning \cite{al2023untargeted} and backdoor attacks \cite{lyu2023poisoning} against FL model aggregation. These attacks commonly use poisoned training data to manipulate the aggregation result. We will show that OFL can adapt defense such as similarity-based methods \cite{10559866,mao2021romoa} smoothly while keeping models encrypted.

\subsection{Resource Optimizing Problem}
We note that various resource-optimizing problems have been defined in related studies \cite{ye2024asteroid,wang2021resource,guo2023privacy}. However, unlike previous work, we formalize the resource-optimizing problem under more strict constraints. Particularly, we consider each client's resource of computing, communication, data heterogeneity, availability, and privacy budget, which are denoted by $\bm{R}=\{R_{comp}, R_{comm}, R_{data}, R_{avai}, R_{priv}\}$. Each kind of resource has a limitation $L_R^t$ in timestep $t$. Assuming that the amortized cost of each kind of resource for training and uploading a local parameter is $C_{R,up}$ while downloading a global parameter is $C_{R,down}$, we set the objective function of a client as to upload and download updated parameters as much as possible, which meets the demand of clients' attendance in FL, contributing to and learning from the global model. To this end, we have
{\small
\begin{equation}
\label{eq:opt}
\begin{aligned}
    &\max (\gamma_{up}^t+\gamma_{down}^t) |\bm{\theta}| \\
    & \textrm{s.t.} \begin{cases}
        \gamma_{up}^t = |\bm{\theta}_{up}^t| / |\bm{\theta}_{global}| \\
        \gamma_{down}^t = |\bm{\theta}_{down}^t| / |\bm{\theta}_{global}| \\
        (C_{R,up} \cdot \gamma_{up}^t + C_{R,down} \cdot \gamma_{down}^t) |\bm{\theta}| \leq L_R^t, \forall R \in \bm{R} \\
        1 \leq t \leq T 
    \end{cases}
\end{aligned}
\end{equation}
}

\section{OFL Framework}
\label{sec:ofl}
\subsection{Design Overview}
The design rationale of opportunistic federated learning (OFL) can be explained from two angles. First, OFL uses a hierarchical aggregation strategy supporting asynchronous and synchronous aggregation in different learning phases. Synchronous aggregation provides the learning convergence guarantee, while asynchronous aggregation enables resource-limited clients to join the global learning opportunistically. Second, OFL integrates a secure solution specifically designed for resource-limited and privacy-aware clients. OFL takes advantage of the parameter selection process to build a lightweight privacy-preserving model uploading scheme in asynchronous aggregation, providing the DP guarantee. Meanwhile, OFL offers strong security guarantees in synchronous aggregation with FHE techniques, ensuring data security and aggregation security at the same time.

Generally, OFL consists of two phases: intra-cluster aggregation and inter-cluster aggregation. As introduced in the system model, the intra-cluster phase aggregates client models within the cluster into a leader model via a selected leader, while the inter-cluster phase aggregates leader models into a global model via a central parameter server (PS).

Before the intra-cluster phase begins, clients are grouped into clusters according to resource states. Each cluster will select a relatively high-performance client as the leader who is responsible for aggregation within the cluster and communication to the PS. The leader holds a cluster model for aggregation. In the inter-cluster phase, the PS aggregates the leader models into a global model and distributes it to all leaders. It is worth noting that the global model is not transmitted to all clients simultaneously because we adopt an asynchronous aggregation in the client-leader layer, where a client only requires the most urgent updates from the leader. 

Figure ~\ref{fig:system} shows the overview of OFL, while a brief description of OFL is given in Algorithm~\ref{alg:strawman}. Please note that we mainly focus on how clients finish model aggregation and omit the transmission phase. In fact, every client should put the selected and perturbed parameters into an upload queue for transmission to the cluster leader. In the download phase, every client sorts the parameter indexes to be synced by staleness and asks the leader to put sorted parameters into a download queue for transmission to the client. As a result, each client can calculate the ratios of the uploaded and downloaded parameters to the total parameters as $\gamma_{up}$ and $\gamma_{down}$, respectively.
\begin{algorithm}[ht]
\caption{A strawman algorithm of OFL}
\label{alg:strawman}
\begin{algorithmic}[1]
\Require maximum training epoch $T_{max}$, re-clustering interval $T_{rc}$, client resource status $RS^{(0)}$, client number $N$, learning rate $\alpha$, local training iteration $T_{lt}$, privacy budget $\epsilon_1$, $\epsilon_2$
\Ensure global model $\bm{\theta}_{global}$

\State $\bm{\theta}_{i} \xleftarrow{\$} \mathcal{N}(0,1), \forall i \in[1,N]$
\State $CS \leftarrow RS^{(0)}$
\State $t \leftarrow 0$
\While{$t < T_{max}$}
  \If{$t \ \textrm{mod} \ T_{rc} == 0$}
    \State $\{ C_i \}_{i=1}^{k} \leftarrow \textrm{Cluster}(CS,k)$
  \EndIf
  \State $t \leftarrow t + 1$
  \For{$i=0$ to $k$}
    \State select leader $CL_{i} \leftarrow \arg \max_{c \in C_i} CS[c]$ 
    \State $(\bar{\bm{\theta}}_{i}^{(t)}, RS^{(t)}_{i}) \leftarrow \textrm{IntraAgg}(C_i, CL_{i},\alpha,T_{lt},\epsilon_1, \epsilon_2)$
  \EndFor
  \State $(\hat{\bm{\theta}}_{global}^{(t)},RS^{(t)}) \leftarrow \textrm{InterAgg}(\{ CL_i,\bar{\bm{\theta}}_{i}^{(t)}, RS^{(t)}_{i}\}_{i=1}^{k})$
  \State $CS \leftarrow CS \cup RS^{(t)}$
\EndWhile
\State \Return $\bm{\theta}_{global} \leftarrow \hat{\bm{\theta}}_{global}^{(T_{max})}$
\end{algorithmic}
\end{algorithm}

\subsection{Intra-Cluster Aggregation}
The opportunistic feature in the intra-cluster phase has two meanings. First, each client will have an opportunity to contribute to the global model when its resource satisfies the criteria. Second, each parameter of the client model may have an opportunity to be synced with the global model. We note that the first client opportunity has been widely studied in client selection work, such as \cite{cho2022towards,li2022pyramidfl}. Therefore, we generally adapt asynchronous aggregation in the client-leader phase of OFL. We further investigate the second opportunistic feature and design an opportunistic factor (o-factor for short) to fully exploit each client's transmission opportunity. For security concerns, we introduce DP to the inner aggregation phase.

We note that DP has been widely used in privacy-preserving model updating \cite{shokri2015privacy,abadi2016deep}. However, we propose a novelty DP mechanism construction by integrating our o-factor into an exponential mechanism, achieving privacy guarantees and learning efficiency with one shot. Our opportunistic aggregation is also communication-efficient. In observations of practical local model training, only partial parameters are updated noticeably, while the rest are negligible. The most intuitive manifestation is that gradients with significant values benefit the global model's convergence more, which can be explained by learning theory since the descent direction of parameters with significant changes indicates the convergence direction of the model \cite{du2019gradient}.

Now, we define the o-factor considering multiple perspectives. First, the gradient absolute value of each parameter can be seen as an important factor for model updating. Second, we count the staleness of each parameter by tracking its last syncing state. Third, we suggest estimating the density of parameter values since FL on heterogeneous data commonly leads to unique heat maps of client models \cite{ezzeldin2023fairfed}. If unique features cannot contribute to the global model timely, they may vanish due to feature scarcity. Given the client model states of the last two global syncing $\bm{\theta}_0$, $\bm{\theta}_1$, and the latest local training $\bm{\theta}^{\prime}$, the o-factor of each parameter $\theta \in \bm{\theta}$ is
\begin{equation}
\label{eq:ofactor}
\textrm{o-factor} = \frac{| \bm{\theta^{\prime}-\bm{\theta_0}}| \cdot \textrm{DensEst}(\bm{\theta}^{\prime}) } { log_2 (1+| \bm{\theta}_0 - \bm{\theta}_1 |))} 
\end{equation}
where $\textrm{DensEst}()$ is a density estimation process with output range $(0,1)$. If the client model has not been synced globally two times, simply set $\bm{\theta}_1 = \bm{\theta}_0 =0$.

Generally, we can assume that parameters with larger o-factors have higher priority when uploaded. However, this may violate the client's data privacy since parameters with larger o-factors carry more confidential information about the private training data. To tackle the problem, we construct a DP updating scheme based on the exponential mechanism \cite{mcsherry2007mechanism}. As we can see, the o-factor is constructed as a natural utility function. With an additional process, we have the utility score $u_{\bm{\theta}} = \textrm{Norm}(\bm{\theta})$, where $\textrm{Norm}()$ is a max-min normalization process. In this case, we can give the sensitivity of the utility function as $\Delta_{u} = \max_{\forall \theta_0,\theta_1 \in \bm{\theta}} | u_{\theta_0} - u_{\theta_1}| $.

Now, each client in OFL can select and upload parameter $\theta$ with probability proportional to $exp(\frac{\epsilon \cdot u_{\theta}}{2\Delta_{u}})$. We note that this selection can preserve clients' data privacy at some level. However, subtle attacks in practices like membership inference or data reconstruction attacks still have odds to win \cite{nasr2019comprehensive,zhao2024loki}. Therefore, we further mitigate threats by introducing an additional Laplace mechanism to the selected parameters. To this end, we split the total privacy budget $\epsilon$ into $\epsilon_1$ and $\epsilon_2$ for parameter selection and uploading, respectively. It is challenging to apply the original Laplace mechanism \cite{dwork2014algorithmic} directly since model parameters have an unbounded variance, which is quite different from clipped gradients \cite{abadi2016deep}. Fortunately, clients in OFL prefer to select parameters with high-density metrics due to the o-factor design. We can perturb the distance between actual parameter values and their centroids to construct a model publishing process with DP \cite{mao2021secure}. Given the selected parameter set $\bm{\theta}_{sel}$, a kernel density estimation algorithm \textrm{KDE} outputs $k$ clusters $\{ PC_{1}, PC_{2}, \ldots, PC_{k} \}$ with centroids $\{ \omega_{1}, \omega_{2}, \ldots \omega_{k} \}$. Then, the sensitivity of each cluster is $\Delta_{PC_{i}} = 2\max_{\theta_0 \in PC_{i}} | \theta_0-\omega_{i} |$. The additional noise will be drawn from $\textrm{Lap}(\Delta_{PC_{i}}/\epsilon_{PC_{i}})$, where $\epsilon_2 = \sum_{i=1}^{k} \epsilon_{PC_{i}}$.

By summarizing the aforementioned implementations, we give the complete intra-cluster aggregation in Algorithm~\ref{alg:intra}. 
\begin{algorithm}[ht]
\caption{IntraAgg: intra-cluster aggregation}
\label{alg:intra}
\begin{algorithmic}[1]
\Require client cluster $C$, cluster leader $CL$, learning rate $\alpha$, local training iteration $T_{lt}$, privacy budget $\epsilon_1$, $\epsilon_2$
\Ensure leader model $\bar{\bm{\theta}}$

\State $\bm{t} \leftarrow \{0\}^{|C|}$
\State \textit{Client:}
\ForAll{client $c$ in $C$}
  \State $\gamma_{up}^t, \gamma_{down}^t \leftarrow \textrm{solve Equation~\ref{eq:opt}}$
  \State $ \{ \bm{g}_{c,i} \}_{i=1}^{\gamma_{up}|\bm{\theta}_c|} \leftarrow \mathcal{L}_{c}(\{\bm{\theta}^{t}_{c,i}\}_{i=1}^{\gamma_{up}|\bm{\theta}_c|};\mathcal{D}_c)$  
  \State $ \bm{\theta}^{(t+1)}_{c} \leftarrow \bm{\theta}^{(t)}_{c} + \alpha \cdot \bm{g}_{c}$
  \State $t \leftarrow t + 1$
  \If{$t \ \textrm{mod} \ T_{lt} == 0$}
    \State $\bm{\theta}_{sel} \xleftarrow{e^{\frac{\epsilon_1 \cdot u_{\bm{\theta}}}{2\Delta_{u}}}} \bm{\theta}^{(t)}_{c}$
    \State $\bm{z}^{|\bm{\theta}_{sel}|} \xleftarrow{\$} \textrm{Lap}(k\Delta_{C}/\epsilon_{2})$
    \State $\bm{\theta}_{sel} \leftarrow \bm{\theta}_{sel} + \bm{z}$
    \State send $\bm{\theta}_{sel}$ and resource status $RS^{(t)}_{c}$ to $CL$
  \EndIf
\EndFor

\State \textit{Leader} $CL$:
\While{\textsf{true}}
  \State receive $\bm{\theta}_{sel}$ and $RS^{(t)}_{c}$ from client $c$
  \State $\bar{\bm{\theta}} \leftarrow \bar{\bm{\theta}} + 1/|C| \cdot \bm{\theta}_{sel}$
  \State $RS_{c} \leftarrow RS_{c} \cup RS^{(t)}_{c}$
\EndWhile
\State \Return $\bar{\bm{\theta}}$
\end{algorithmic}
\end{algorithm}

\subsection{Inter-Cluster Aggregation}
The inter-cluster phase consists of two main steps: inter-cluster aggregation and re-clustering. Inter-cluster aggregation basically mirrors the original FL synchronous aggregation. For security concerns, we exploit a threshold FHE to protect leader models from the semi-honest PS. We choose FHE to construct secure aggregation rather than MPC because OFL also considers model aggregation security. An efficient model poisoning detection technique is developed, sanitizing the poisoned model through frequency analysis on ciphertexts. The re-clustering process groups clients with similar characteristics periodically. The features for clustering include device specifics, resource status, and model states.

Different from classical FHE schemes, a threshold FHE (ThFHE) scheme \cite{boneh2018threshold} allows multiple parties holding secret key shares to decrypt ciphertext partially, which is a promising tool to outsource computation for multiple clients with private data. Recently, a low-communication ThFHE \cite{passelegue2025low} is designed to provide an IND-CPA security guarantee with an additional procedure adding crucial randomness. In OFL, we alter the ThFHE scheme to fit our security requirement, i.e., each client's data security against the semi-honest PS and semi-honest clients. We note that the data security guarantee of OFL is stronger than the existing work from two perspectives: (1) OFL is resistant to the collusion of the semi-honest PS and clients, and (2) OFL supports a lightweight but accurate analysis method to detect and sanitize model poisoning attacks.


The most common defense against poisoning attacks is to detect outliers during the aggregation. However, it is rather challenging to detect the encrypted models efficiently. Inspired by a backdoor detection method in deep neural networks \cite{zeng2021rethinking}, we find the feasibility of detecting malicious behaviors in model aggregation through frequency analysis. Previous robust FL aggregation methods commonly use distance or similarity measurements for anomaly detection \cite{zhao2021sear,mao2021romoa,xu2025dual}, which suffers from intensive computation due to pairwise comparisons and performs unsatisfactorily in heterogeneous data settings. Considering the data heterogeneity and resource limitation, we propose a frequency analysis-aided secure model aggregation method in the inter-cluster phase.

In OFL, all parameters of each leader model are transformed into the frequency domain via fast Fourier transformation (FFT). In the frequency domain, the PS can easily tell high-frequency components in each model. According to studies in backdoor detection \cite{zeng2021rethinking}, high-frequency components are highly relevant to poisoning attacks. Therefore, model parameters with high frequency will be seen as suspicious. Please note this detection is done solely, without comparison with other models. Therefore, our method can be implemented in $\mathcal{O}(\log_{2} |\bm{\theta}|)$ complexity based on the technique in \cite{cheon2018faster}. To further consider the correlation between models, we use the Manhattan distance between the global model and each leader model to detect relative outliers. By combining the frequency analysis result with the Manhattan distance measurement, we design a sanitizing factor for secure aggregation as
\begin{equation}
\label{eq:sfactor}
\begin{split}
\llbracket \textrm{s-factor}_{i} \rrbracket & = 1-  MagNorm( \textsf{FHE-DFT}(\llbracket \bar{\bm{\theta}}_{i} \rrbracket)) \\ 
& \times DistNorm(\textsf{FHE-MD}(\llbracket \bar{\bm{\theta}}_{i} \rrbracket, \llbracket {\bm{\theta}}_{global} \rrbracket)),
\end{split}
\end{equation}
where \textsf{FHE-DFT} is a discrete Fourier transform algorithm in FHE ciphertext space, and its implementation is given in Algorithm~\ref{alg:dft}. The $MagNorm()$ is a max-min normalization process of frequency magnitude. \textsf{FHE-MD} is a Manhattan distance algorithm in FHE ciphertext space, i.e., $\textsf{FHE-MD}(\llbracket \bar{\bm{\theta}}_{i} \rrbracket$, $\llbracket {\bm{\theta}}_{global} \rrbracket) = \sum_{j=1}^{| \bar{\bm{\theta}}_i|} | \llbracket \bar{\theta}_{i,j} \rrbracket - \llbracket {\theta}_{global,j} \rrbracket |$. The $DistNorm()$ method is a max-min normalization process of distances.
\begin{algorithm}
\caption{FHE-DFT}
\label{alg:dft}
\begin{algorithmic}[1]
\Require cluster leader $CL_{i}$'s encrypted model $\llbracket \bar{\bm{\theta}}_{i} \rrbracket$
\Ensure an encrypted model in the frequency domain
\State $n \leftarrow |\bar{\bm{\theta}}_i |$
\For{$i=0$ to $n$}
  \State $\omega_n \leftarrow e^{2\pi i/n}$
  \State $\bm{W}_{n/2} \leftarrow \textrm{diag}(1,\omega_n,\omega_n^{2},\ldots,\omega_n^{n/2-1})$
  \State $\bm{D}_{k}= [[\bm{I}_{n/k},\bm{I}_{n/k}], [\bm{W}_{n/k},\bm{W}_{n/k}]] \bm{I}_{n}$
\EndFor
\For{$i=1$ to $\log_2n$}
  \State $\llbracket \tilde{\bm{\theta}}_{0} \rrbracket \leftarrow \textsf{CMult}(\textrm{diag}_0(\bm{D}_{2^i}),\llbracket \tilde{\bm{\theta}} \rrbracket)$
  \State $\llbracket \tilde{\bm{\theta}}_{1} \rrbracket \leftarrow \textsf{LeftRotate}(\llbracket \tilde{\bm{\theta}} \rrbracket, n/2^{i})$
  \State $\llbracket \tilde{\bm{\theta}}_{2} \rrbracket \leftarrow \textsf{RightRotate}(\llbracket \tilde{\bm{\theta}} \rrbracket, n/2^{i})$
  \State $\llbracket \tilde{\bm{\theta}}_{1} \rrbracket \leftarrow \textsf{CMult}(\textrm{diag}_{n/2^{i}}(\bm{D}_{2^i}),\llbracket \tilde{\bm{\theta}}_{1} \rrbracket)$
  \State $\llbracket \tilde{\bm{\theta}}_{2} \rrbracket \leftarrow \textsf{CMult}(\textrm{diag}_{n-n/2^{i}}(\bm{D}_{2^i}),\llbracket \tilde{\bm{\theta}}_{2} \rrbracket)$
  \State $\llbracket \tilde{\bm{\theta}} \rrbracket \leftarrow \textsf{Add}(\llbracket \tilde{\bm{\theta}}_{0} \rrbracket,\llbracket \tilde{\bm{\theta}}_{1} \rrbracket, \llbracket \tilde{\bm{\theta}}_{2} \rrbracket)$
\EndFor
\end{algorithmic}
\end{algorithm}

We now give the construction of the main FHE algorithms used in OFL.
\begin{itemize}
    \item \textsf{KeyGen}($1^{\lambda},1^{N}$)\footnote{We note that FHE keys can be generated and distributed by a trusted CA or via multi-party computation protocol, which is out of our discussion.}: Given the security parameter $\lambda$, let $q_i = p^{i}$, $i \in [1,L]$, $q_{L} = pq_{dec}$ with $q_{dec} \ll q_{L}$, and choose a power-of-two integer $N$, an integer $h$, an integer $P > q_{L}$, and a real number $\sigma >0$ to achieve $\lambda$-bit security level. Sample $s \leftarrow \mathcal{HWT}(h)$, $a \leftarrow \mathcal{R}_{q_{L}}$, and $e \leftarrow \mathcal{DG}(\sigma^{2})$, where $\mathcal{HWT}(h)$ is the set of signed binary vectors in $\{ 0, \pm1 \}^{N}$ with weight $h$ and $\mathcal{DG}(\sigma^{2})$ samples independently from the discrete Gaussian distribution of variance $\sigma^{2}$. Set the secret key as $\textsf{sk} \leftarrow (1,s)$ and the public key as $\textsf{pk} \leftarrow (b,a) \in \mathcal{R}_{q_{L}}^{2}$, where $b \leftarrow -a \cdot s+e \in \mathcal{R}_{q_{L}}$. Sample secret key shares $(\textsf{sk}_{1}, \textsf{sk}_{2},\ldots,\textsf{sk}_{n-1}) \leftarrow \mathbb{Z}_{q_{dec}}^{n-1}$ so that $\textsf{sk} = \textsf{sk}_{1}+\textsf{sk}_{2}+\cdots+\textsf{sk}_{n} \ \textrm{mod} \ q_{dec}$.
    \item \textsf{Enc}(\textsf{pk},m): Encode $m$ by following the encoding process given in \cite{cheon2017homomorphic} and get $m(x) \leftarrow \textsf{Encode}(x)$. Sample $v \leftarrow \mathcal{ZO}(\rho)$ and $e_1,e_2 \leftarrow \mathcal{DG}(\sigma^{2})$, where $\mathcal{ZO}(\rho)$ samples a vector from $\{ 0, \pm1 \}^{N}$ with probability $\rho/2$ of $\pm1$ and probability $1-\rho$ of $0$, and return $\textsf{ct}^{(1,2)} \leftarrow (\lfloor p^{k} \cdot m(x) \rceil + v \cdot \textsf{pk}_a + e_1 \ \textrm{mod} \ q_{L}, v \cdot \textsf{pk}_b + e_2 \ \textrm{mod} \ q_{L})$.
    \item \textsf{ServerDec}(\textsf{pk},\textsf{ct}): Get a fresh ciphertext $\textsf{ct}_{fresh} \leftarrow \textsf{Enc}_{\textsf{pk}}(0)$ and add it to the input ciphertext, $\textsf{ct} \leftarrow \textsf{ct}+\textsf{ct}_{fresh}$. Sample flood noise $\epsilon \leftarrow \mathcal{D}_{\sigma_{flood}}$ and add it to the input ciphertext $\textsf{ct}^{(1)} \leftarrow \textsf{ct}^{(1)} + \epsilon$. Return the server decrypted ciphertext $\textsf{ct}_{dec}^{(1,2)} \leftarrow (\lfloor 1/p \cdot \textsf{ct}^{(1)} \rceil_{\sigma_1} \ \textrm{mod} \ q_{dec}$, $\lfloor 1/p \cdot \textsf{ct}^{(2)} \rceil_{\sigma_2} \ \textrm{mod} \ q_{dec})$.
    \item \textsf{PartDec}($\textsf{sk}_i$,$\textsf{ct}_{dec}$): Sample noise $z_i \leftarrow \mathcal{D}_{\mathbb{Z},\eta}$ and add it to the server decrypted ciphertext. Return $\textsf{pd}_{i} \leftarrow \textsf{ct}_{dec}^{(1)} \cdot \textsf{sk}_{i} + z_i \ \textrm{mod} \ q_{dec}$.
    \item \textsf{FinDec}($\textsf{ct}_{dec}$, $\textsf{pd}_1,\textsf{pd}_2,\ldots,\textsf{pd}_N$): Aggregate partial decryption shares $\textsf{pt} \leftarrow \sum_{i=1}^{n} \textsf{pd}_{i} + \textsf{ct}_{dec}^{(2)} \ \textrm{mod} \ q_{dec}$. Return final decryption result after decoding as $m^{\prime} \leftarrow \textsf{Decode}(\textsf{pt})$.
\end{itemize}

As for the global clients' re-clustering process, it is the PS's job to collect clients' status reports and form new clusters for active clients. As mentioned before, OFL takes into account multiple factors for clustering, including device specifics, resource status, and model states. The server will keep tracking clients' status reports and run the re-clustering process periodically. Since clustering algorithms have been well studied, we generally apply dynamic Kmeans in OFL. We summarize essential steps of the global aggregation in Algorithm~\ref{alg:inter}.
\begin{algorithm}[ht]
\caption{InterAgg: inter-cluster aggregation}
\label{alg:inter}
\begin{algorithmic}[1]
\Require cluster leader, leader model, and resource status of each cluster $\{ CL_i,\bar{\bm{\theta}}_{i}^{(t)}, RS^{(t)}_{i}\}_{i=1}^{k}$
\Ensure global model parameters $\bm{\theta}_{global}$
\State initialize $t \leftarrow 0$
\State \textit{Cluster Leader $CL$}:
\For{$CL_{i}$, $i \in[1,k]$}
  \State $\llbracket \bar{\bm{\theta}}_i \rrbracket \leftarrow \textsf{Enc}(\textsf{pk},\bar{\bm{\theta}}_i)$
  \State send $\llbracket \bar{\bm{\theta}}_i \rrbracket$ and $\bm{RS}_{i}$ to the server \Comment{wait for server}
  \State receive $\textsf{ct}_{dec}$ from the server
  \State $\textsf{pd}_{i} \leftarrow \textsf{PartDec}(\textsf{sk}_i,\textsf{ct}_{dec})$
  \State send $\textsf{pd}_{i}$ to the server
\EndFor
\State \textit{Server}:
\State collect $\llbracket \bar{\bm{\theta}}_i \rrbracket$ and $\bm{RS}_{i}$ from $\{ CL_{i} \}_{i=1}^{k}$, \ $t \leftarrow t+1$
\State calculate s-factor using Equation~\ref{eq:sfactor}
\State $\llbracket \bm{\theta}_{global} \rrbracket \leftarrow 1/n \sum_{i=1}^{n} \textrm{s-factor}_{i} \cdot \llbracket \bar{\bm{\theta}}_i \rrbracket$
\State $\textsf{ct}_{dec}$ $\leftarrow$ \textsf{ServerDec}(\textsf{pk},$\llbracket \bm{\theta}_{global} \rrbracket$)
\State send $\textsf{ct}_{dec}$ to $\{ CL_{i} \}_{i=1}^{k}$ \Comment{wait for $CL$s}
\State collect $\textsf{pd}_{i}$ from $\{ CL_{i} \}_{i=1}^{k}$
\State $\hat{\bm{\theta}}_{global} \leftarrow \textsf{FinDec}(\textsf{ct}_{dec}, \textsf{pd}_1,\textsf{pd}_2,\ldots,\textsf{pd}_k)$
\State update client status records $\bm{CS} \leftarrow \bm{CS} \cup \{ \bm{RS}_{i} \}_{i=1}^{k}$
\State \Return $\bm{\theta}_{global} \leftarrow \hat{\bm{\theta}}_{global}$
\end{algorithmic}
\end{algorithm}

\section{Theoretical Analysis}
\label{analysis}
\subsection{Security Analysis}
The security guarantee of OFL stands on the security of two phases, i.e., intra-cluster aggregation and inter-cluster aggregation, which will be discussed separately. The conclusion is that intra-cluster aggregation satisfies $\epsilon$-DP, while inter-cluster aggregation has chosen-plaintext attack (CPA) security.

\subsubsection{Intra-Cluster Aggregation}
To precisely analyze the security guarantee, we divide Algorithm~\ref{alg:intra} into two parts: selection and perturbation. Supposing that $\epsilon$ is the privacy budget for the entire algorithm, the budget can be divided into $\gamma_{up}|\bm{\theta}|\epsilon_1$ and $\gamma_{up} |\bm{\theta}| \epsilon_2$, $\epsilon = \gamma_{up}|\bm{\theta}|\epsilon_1 + \gamma_{up} |\bm{\theta}| \epsilon_2$. Then we have the following two lemmas:
\begin{lemma} \label{lemma1}
The selection process is $ \gamma_{up}|\bm{\theta}|\epsilon_1$-DP.
\end{lemma}
\begin{lemma} \label{lemma2}
The perturbation process is $ \gamma_{up} |\bm{\theta}|  \epsilon_2$-DP.
\end{lemma}

\begin{proof}[Proof of Lemma \ref{lemma1}]
Given any neighboring datasets $\mathcal{D}$ and $\mathcal{D}^{\prime}$, the selection process is executed consecutively $ \gamma_{up} |\bm{\theta}|$ times, generating a sequence whose elements constitute the selected parameter set. The selection process using the exponential mechanism satisfies $\epsilon_1-DP$. Thus, by simple composition theorems \citep{dwork2014algorithmic}, we have that the selection process satisfies $ \gamma_{up}|\bm{\theta}|\epsilon_1$-DP. 
\end{proof}

\begin{proof}[Proof of Lemma \ref{lemma2}]
The perturbation process may be affected by density estimation. Since parameter updates in each iteration can be bounded by clipping and learning rate, we assume that the clusters yielded by density estimation remain unchanged in adjacent iterations. Then, $\Delta f$ of the model parameter in cluster $i$ has $\Delta f_i \le  2\max_{\theta' \in PC_i}|\theta' - \omega_i| = \Delta_{PC_i}$. The perturbation process executes The Laplace mechanism on each element in the selected parameter set with $\Delta f=\Delta _{PC_i}$ and ${\epsilon_2\over k} < \epsilon_2$, and since the $\Delta f$ of all parameters within $PC_i$ have limitation $\Delta_{PC_i}$, the perturbation process satisfies $\epsilon_2-DP$ for any single parameter. The selected parameter set has at most $ \gamma_{up} |\bm{\theta}|$ elements; therefore, by simple composition theorems, we find that the perturbation process satisfies $ \gamma_{up}|\bm{\theta}|\epsilon_2$-DP.
\end{proof}

\subsubsection{Inter-Cluster Aggregation}
We will show that the server and cluster leaders in OFL provide CPA security in the inter-cluster aggregation phase. Recall that the server and cluster leaders are considered semi-honest, faithfully executing their jobs while potentially being curious about confidential data. We omit the evaluation key $\textsf{ek}$ in the analysis and then obtain our security claims by additionally assuming that the security still holds even when the supplementary information contained in $\textsf{ek}$ is disclosed. This aligns with the standard circular security assumption. We begin with the definition of OW-CPA security as presented in \cite{passelegue2025low}.
\begin{definition}[OW-CPA security for ThFHE] We say that a threshold FHE scheme ThFHE = (\textsf{KeyGen}, \textsf{Enc}, \textsf{Eval}, \textsf{ServerDec}, \textsf{PartDec}, \textsf{FinDec}) is $Q_D-OW-CPA$ secure, if for all PPT adversaries $\mathcal{A} = (\mathcal{A}_0, \mathcal{A}_1)$ with query bound $Q_D$, we have $Adv_{ThFHE}^{Q_D-OW-CPA}(\mathcal{A}) := Pr[Expt_{\mathcal{A}, ThFHE}^{OW-CPA} (1^{\lambda}, 1^N) = 1] = negl(\lambda)$, where $Exp_{\mathcal{A}, ThFHE}^{OW-CPA}$ is the experiment in Figure~\ref{fig:OW-CPA}. We abbreviate it as OW-CPA security in the rest.
\end{definition}


\begin{figure}[ht]
\begin{minipage}{\columnwidth}
    \hrulefill
    \begin{algorithmic}
        \State \textbf{Exp}$_{ThFHE}^{OW-CPA}$:
        \State 1: ($\textsf{pk}, \textsf{ek}, \textsf{sk}_1, ..., \text{sk}_N) \leftarrow \textsf{KeyGen}(1^{\lambda}, 1^N$)
        \State 2: $\textsf{ctr} \leftarrow 0, \textsf{L} \leftarrow \phi $
        \State 3: $(\textsf{S}, \textsf{st}) \leftarrow \mathcal{A}_0(\textsf{pk})$ with $\textsf{S} \subset [N]$ and $|\textsf{S}| < N $
        \State 4: $\textsf{m} \leftarrow \mathcal{U} (\mathcal{M})$
        \State 5: $(b, \textsf{m}') \leftarrow \mathcal{A}_1^{\textbf{OEnc}, \textbf{OChalb}, \textbf{OEval}, \textbf{ODec}}(\textsf{pk}, \textsf{sk}_i(i \in \textsf{S}), \textsf{st})$
        \State \textbf{return} $\textsf{m}' = \textsf{m} \land b = 1$
    \end{algorithmic}
    \hrulefill
    \begin{algorithmic}
        \State \textbf{OEnc}($\textsf{m}$):
        \State 1: $\textsf{ct} \leftarrow \textsf{Enc}(\textsf{pk}, \textsf{m})$
        \State 2: $\textsf{ctr} \leftarrow \textsf{ctr} + 1$
        \State 3: $\textsf{L}[\textsf{ctr}] \leftarrow (0, \textsf{m}, \textsf{ct})$ 
        \State \textbf{return} $\textsf{ct}$
    \end{algorithmic}
    \vspace{3pt}
    \begin{algorithmic}
        \State \textbf{OChall}():
        \State 1: $\textsf{m} \leftarrow \mathcal{U} (\mathcal{M})$
        \State 2: $ \textsf{ct} \leftarrow \textsf{Enc}(\textsf{pk}, \textsf{m})$ 
        \State 3: $ \textsf{ctr} \leftarrow \textsf{ctr} + 1 $
        \State 4: $ \textsf{L}[\textsf{ctr}] \leftarrow (1, \textsf{m}, \textsf{ct}) $
        \State \textbf{return} $\textsf{ct}$
    \end{algorithmic}
    \vspace{3pt}
    \begin{algorithmic}
        \State \textbf{OEval}$(\textsf{fun}, (i_1, . . . , i_l))$:
        \State 1: \textbf{For} $j \in [l]$: $(b_j, \textsf{m}_j, \textsf{ct}_j) \leftarrow \textsf{L}[i_j]$
        \State 2: \textsf{ct} $\leftarrow$ \textsf{Eval}($\textsf{ek}, \textsf{fun}, \textsf{ct}_1, . . . , \textsf{ct}_l$)
        \State 3: $\textsf{m} \leftarrow \textsf{fun}(\textsf{m}_1, . . . , \textsf{m}_l)$
        \State 4: $b \leftarrow b_1 \lor ... \lor b_l$
        \State 5: $\textsf{ctr} \leftarrow \textsf{ctr} + 1 $
        \State 6: $\textsf{L}[\textsf{ctr}] \leftarrow (\textsf{m}, \textsf{ct})$
        \State \textbf{return} $\textsf{ct}$
    \end{algorithmic}
    \vspace{3pt}
    \begin{algorithmic}
        \State \textbf{ODec}$(j)$:
        \State 1: \textbf{If} $j > \textsf{ctr}$: \textbf{return} $\perp$ 
        \State 2: $(b, \textsf{m}, \textsf{ct}) \leftarrow \textsf{L}[j]$
        \State 3: $\textsf{ct}_{dec} \leftarrow \textsf{ServerDec}(\textsf{pk}, \textsf{ct})$ 
        \State 4: $\textsf{pd}_k \leftarrow \textsf{PartDec}(\textsf{sk}_k, \textsf{ct}_{dec})$, for $k \in [N]$
        \State \textbf{return} $(\textsf{ct}_{dec}, (\textsf{pd}_k)_{k\in[N]})$
    \end{algorithmic}
    \hrulefill
\end{minipage}
\caption{$Q_D$-OW-CPA security game for ThFHE.}
\label{fig:OW-CPA}
\end{figure}

\begin{definition}[IND-CPA secure for ThFHE] We say that a threshold FHE scheme ThFHE = (\textsf{KeyGen}, \textsf{Enc}, \textsf{Eval}, \textsf{ServerDec}, \textsf{PartDec}, \textsf{FinDec}) is $Q_D-IND-CPA$ secure if for all PPT adversaries $\mathcal{A} = (\mathcal{A}_0, \mathcal{A}_1)$ making at most $Q_D$ decryption queries, we have $Pr[A(Exp^{IND-CPA}_1(1^{\lambda}, 1^N)) = 1] - Pr[A$ $(Exp^{IND-CPA}_0(1^{\lambda}, 1^N )) = 1] \leq negl(\lambda)$, where the experiment is described in Figure~\ref{fig:IND-CPA}.
\end{definition}


\begin{figure}[ht]
\begin{minipage}{\columnwidth}
    \hrulefill
    \begin{algorithmic}
        \State \textbf{Exp}$^{Th-IND-CPA}_{b}(1^{\lambda}, 1^N)$:
        \State 1: $(\textsf{pk}, \textsf{ek}, \textsf{sk}_1, ..., \textsf{sk}_N ) \leftarrow \textsf{KeyGen}(1^{\lambda}, 1^N )$
        \State 2: $\textsf{ctr} \leftarrow 0, \textsf{L} \leftarrow \phi $
        \State 3: $(\textsf{S}, \textsf{st}) \leftarrow \mathcal{A}_0(\textsf{pk})$ with $\textsf{S} \subset [N]$ and $|\textsf{S}| < N $
        \State 4: $b \leftarrow \mathcal{U} ({0, 1})$
        \State 5: $b' \leftarrow \mathcal{A}_1^{\textbf{OEnc}, \textbf{OChalb}, \textbf{OEval}, \textbf{ODec}}(\textsf{pk}, \textsf{sk}_i(i \in \textsf{S}), \textsf{st})$
        \State \textbf{return} $b' = b$
    \end{algorithmic}
    \hrulefill
    \begin{algorithmic}
        \State \textbf{OEnc}$(\textsf{m})$:
        \State 1: $\textsf{ct} \leftarrow \textsf{Enc}(\textsf{pk}, \textsf{m})$
        \State 2: $\textsf{ctr} \leftarrow \textsf{ctr} + 1$
        \State 3: $\textsf{L}[\textsf{ctr}] \leftarrow (\textsf{m}, \textsf{m}, \textsf{ct})$ 
        \State \textbf{return} $\textsf{ct}$
    \end{algorithmic}
    \vspace{3pt}
    \begin{algorithmic}
        \State \textbf{OChall}$_b(\textsf{m}_0, \textsf{m}_1)$:
        \State 1: \textbf{If} $|\textsf{m}_0| \neq |\textsf{m}_1|$: \textbf{return} $\perp$
        \State 2: $ \textsf{ct} \leftarrow \textsf{Enc}(\textsf{pk}, \textsf{m}_b)$ 
        \State 3: $ \textsf{ctr} \leftarrow \textsf{ctr} + 1 $
        \State 4: $ \textsf{L}[\textsf{ctr}] \leftarrow (\textsf{m}_0, \textsf{m}_1, \textsf{ct}) $
        \State \textbf{return} $\textsf{ct}$
    \end{algorithmic}
    \vspace{3pt}
    \begin{algorithmic}
        \State \textbf{OEval}$(\textsf{fun}, (i_1, . . . , i_l))$:
        \State 1: \textbf{For} $j \in [l]: (\textsf{m}_{0,j}, \textsf{m}_{1,j}, \textsf{ct}_j ) \leftarrow \textsf{L}[i_j]$
        \State 2: $\textsf{ct} \leftarrow \textsf{Eval}(\textsf{ek}, \textsf{fun}, \textsf{ct}_1, . . . , \textsf{ct}_l)$
        \State 3: $\textsf{ctr} \leftarrow \textsf{ctr} + 1 $
        \State 4: $\textsf{m}_0 \leftarrow \textsf{fun}(\textsf{m}_{0,1}, . . . , \textsf{m}_{0,l})$
        \State 5: $\textsf{m}_1 \leftarrow \textsf{fun}(\textsf{m}_{1,1}, . . . , \textsf{m}_{1,l})$
        \State 6: $\textsf{L}[\textsf{ctr}] \leftarrow (\textsf{m}_0, \textsf{m}_1, \textsf{ct})$
        \State \textbf{return} $\textsf{ct}$
    \end{algorithmic}
    \vspace{3pt}
    \begin{algorithmic}
        \State \textbf{ODec}$(j)$:
        \State 1: \textbf{If} $j > \textsf{ctr}$: \textbf{return} $\perp$ 
        \State 2: $(\textsf{m}_0, \textsf{m}_1, \textsf{ct}) \leftarrow \textsf{L}[j]$
        \State 3: \textbf{If} $\textsf{m}_0 \neq \textsf{m}_1$: \textbf{return} $\perp$  
        \State 4: $\textsf{ct}_{dec} \leftarrow \textsf{ServerDec}(\textsf{pk}, \textsf{ct})$ 
        \State 5: $\textsf{pd}_k \leftarrow \textsf{PartDec}(\textsf{sk}_k, \textsf{ct}_{dec})$, for $k \in [N]$
        \State \textbf{return} $(\textsf{ct}_{dec}, (\textsf{pd}_k)_{k\in[N]})$
    \end{algorithmic}
    \hrulefill
\end{minipage}
\caption{$Q_D$-IND-CPA security game for ThFHE.}
\label{fig:IND-CPA}
\end{figure}

\begin{lemma} \label{lemma3} Let ThFHE be an OW-CPA secure threshold FHE scheme. Then Algorithm~\ref{alg:inter} is OW-CPA secure for the server and the cluster leaders.
\end{lemma}
\begin{proof}[Proof of Lemma \ref{lemma3}]
Recall that the ThFHE scheme in \cite{passelegue2025low} satisfies the OW-CPA security. We omit the detailed proof here. During inter-cluster aggregation, the server and cluster leaders share $\llbracket \bar{\bm{\theta}}_i \rrbracket$, $\textsf{ct}_{dec}$, $\textsf{pd}_i$ and $RS_i$, while $RS_i$ does not reveal any additional information about the plaintext or secret key. Consequently, neither the server nor any cluster leader can gain any advantage about $\llbracket \bar{\bm{\theta}}_i \rrbracket$, $\textsf{ct}_{dec}$, and $\textsf{pd}_i$ in winning the OW-CPA game. Therefore, Algorithm~\ref{alg:inter} satisfies the OW-CPA security.
\end{proof}

\begin{lemma} \label{lemma4} Let ThFHE be an IND-CPA secure threshold FHE scheme. Then Algorithm~\ref{alg:inter} is Th-IND-CPA secure for cluster leaders.
\end{lemma}
\begin{proof}[Proof of Lemma \ref{lemma4}] To achieve IND-CPA security, additional randomness has been introduced by \textsf{ServerDec} operation. In Algorithm~\ref{alg:inter}, the messages transmitted and received by cluster leaders cannot provide additional information. Therefore, Algorithm~\ref{alg:inter} achieves IND-CPA security for cluster leaders.

Please note that the server performs \textsf{ServerDec} and \textsf{FinDec} in our scheme, granting it access to additional information. Given that the server is semi-honest, we assert that \textsf{ServerDec} remains functionally valid. If server is accessible to $\llbracket \bar{\bm{\theta}}_{global} \rrbracket$ and $\textsf{pd}_i$, the security will be degenerated. We allow the server to collude with some cluster leaders, thereby degenerating the security to OW-CPA security. However, we note that if an independent server is introduced for \textsf{ServerDec}, the PC can inherit the IND-CPA security in Algorithm~\ref{alg:inter}.
\end{proof}

\subsection{Complexity Analysis}
We will give the breakdown costs of each role in one round of OFL global training as interpreted in the Algorithm~\ref{alg:strawman}. For simplicity, we assume that $k$ clusters and leaders are fixed in the analysis and all clients use the same uploading and downloading rates, $\gamma_{up}$, $\gamma_{down}$, $k=\log_{2}N$. As for the FHE part, we average the costs of all FHE operations (including encryption, decryption, and evaluation) on input size $n$ as $\mathcal{O}(n)$, which is commonly used in related work.

If a client in OFL is not selected as a cluster leader, then its costs are optimistic. Recall algorithms of OFL, the client role undertakes three main jobs: local training, opportunistic uploading and downloading parameters. The perturbation process in opportunistic uploading mainly involves o-factor calculating and random sampling. For density estimation, we use the kernel output and its activation results as an approximation, which completely reuses the forward propagation of local training without any cost. Meanwhile, the uploading and downloading processes may be limited by communication resources. Therefore, the time complexity of each client is $\mathcal{O}(T_{lt} \cdot |\bm{\theta}|)$, while the communication complexity is $\mathcal{O}((\gamma_{up}+\gamma_{down}) \cdot |\bm{\theta}| )$. Both complexity expressions are only relevant to the model size and irrelevant to the client scale.

When a client becomes a cluster leader, there will be costs in addition to its client job. Each cluster leader mainly involves four steps: local model aggregation, uploading the encrypted leader model, partial decryption, and downloading the global model. We simply assume that the leader model encryption introduces $\mathcal{O}(|\bm{\theta}|)$ FHE encrypting operations without any optimizations like cipher packing. The partial decryption involves the same amount of FHE decrypting operations with additional upload and download of the partially decrypted model. At the end of each global synchronization, the global model should be forwarded to clients. Thus, the time complexity of a cluster leader is $\mathcal{O}( (\gamma_{up} \cdot T_{lt} \cdot |C| +2) \cdot |\bm{\theta}|)$. The communication complexity is $\mathcal{O}((\gamma_{up} \cdot T_{lt} \cdot |C| + \gamma_{down}\cdot |C| + 4) \cdot |\bm{\theta}|)$. When we set $k=\log_{2}N$, $|C| \leq N/\log_{2}N$.

The server undertakes the most aggregation work in OFL, including encrypted model aggregation and partial decryption aggregation. The client clustering process has a $\mathcal{O}(kN)$ time complexity. The encrypted model aggregation mainly involves the s-factor calculation and a weighted sum. The $\textsf{FHE-DFT}$ and $\textsf{FHE-MD}$ processes have about $5k \cdot \log_{2}|\bm{\theta}|$ FHE multiplications and $(k \cdot (|\bm{\theta}| +  \log_{2}|\bm{\theta}|)$ FHE additions. A naive weighted sum process has $k|\bm{\theta}|$ FHE multiplications and $k|\bm{\theta}|$ additions. The $\textsf{ServerDec}$ process involves $|\bm{\theta}|$ encryptions, $2|\bm{\theta}|$ additions and $|\bm{\theta}|$ decryptions. Thus, the server's time complexity is about $\mathcal{O}((3k+4) \cdot |\bm{\theta}| + 6k \cdot \log_{2}|\bm{\theta}|)$. The communication complexity is about $\mathcal{O}(4k \cdot |\bm{\theta}|)$. Since $k=\log_{2}N$, both complexity expressions are logarithmic to the client scale and linear to the model size, indicating OFL's good scalability.

\section{Evaluation}
\subsection{Implementation and Setup}
\label{sec: experimental setup}
We have implemented an OFL prototype on a real-world testbed, as illustrated in Figure~\ref{fig:ofl-realworld}. This testbed comprises 16 Jetson Xavier NX development boards as resource-limited devices and a high-performance server. All devices are connected to a LAN through a switch. To mimic heterogeneous devices, we restrict the capability of development boards by adjusting operating frequency and available units. Table~\ref{tab:setting} gives four types of heterogeneous devices\footnote{The device has a higher maximal frequency when 2 cores are online than that of 4 or more CPU cores.} used for evaluation. The central server has 8 Intel Core i7-9700 CPUs and 32GB RAM.
\begin{figure}[ht]
\centering
\includegraphics[width=\linewidth]{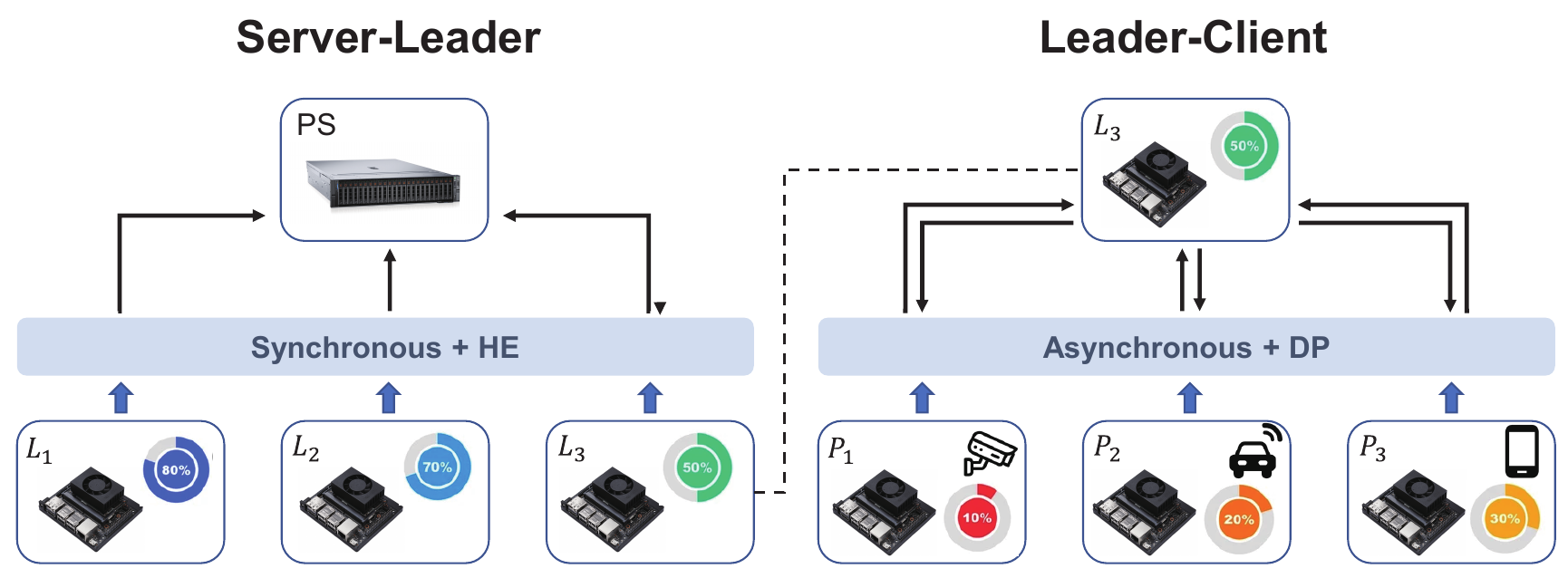}
\caption{Real-world testbed for OFL evaluation.}
\label{fig:ofl-realworld}
\end{figure}

\begin{table}[ht]
\centering
\caption{Four types of heterogeneous devices.}
\label{tab:setting}
{\scriptsize
\begin{tabular}{ccccc}
\toprule
Type & \textbf{A} & \textbf{B} & \textbf{C} & \textbf{D} \\
\midrule
Power & 10W & 10W & 15W & 20W \\
CPU Cores & 2 & 4 & 4 & 6 \\
Max CPU Freq. (GHz) & 1.497 & 1.190 & 1.420 & 1.420 \\
Max GPU Freq. (MHz) & 803.25 & 803.25 & 1,109.25 & 1,109.25 \\
Max EMC Freq. (MHz) & 1,600 & 1,600 & 1,600 & 1,866 \\
\bottomrule
\end{tabular}
}
\end{table}

We evaluate OFL on different tasks. Image classification tasks use MNIST \cite{deng2012mnist}, FMNIST \cite{xiao2017/online}, and SVHN \cite{netzer2011reading} datasets for training CNN models with two 784-neuron convolutional layers and one 1024-neuron fully connected layer. Graph classification tasks use Cora \cite{mccallum2000automating}, Citeseer \cite{giles1998citeseer}, and PubMed \cite{sen2008collective} datasets for training models with a 2-layer GCN and one 512-neuron fully connected layer. In the data preprocessing, each image dataset is divided into multiple subsets independently to form non-i.i.d. datasets using Dirichlet distributions \cite{9835537}, while each graph dataset uses data enhancement techniques to produce diverse datasets. The enhancement is implemented by combining the graph splitting method and contrastive learning \cite{you2020graph}.

We evaluate OFL specifically on two critical aspects: model performance and resource efficiency. Our evaluation aims to provide a deep understanding of how OFL performs in resource-heterogeneous and privacy-aware cases. Since security guarantees can be proved theoretically, we choose three frameworks for performance comparison, i.e., the original FL\cite{mcmahan2017communication}, CFL \cite{sattler2020clustered}, and REFL \cite{wang2021resource}. We have also implemented a simulation based evaluation for large-scale FL cases, investigating OFL performance and its scalability. This simulation is conducted on the high-performance server with a lightweight thread simulating each client.

\subsubsection{Training Performance}
We first give an overall result of training performance in Table~\ref{tab:utility}. The model accuracy of OFL surpasses the original FL baseline on the three graph datasets (Cora, Citeseer, and PubMed). Across five out of the six datasets excluding FMNIST, OFL either matches or exceeds the performance of CFL. More specifically, OFL demonstrates significantly higher accuracy compared to both CFL and REFL on the graph datasets. On the image datasets (MNIST, FMNIST, and SVHN), OFL is highly competitive with the Original FL, outperforming CFL and REFL in most cases. Figure~\ref{fig:ofl-relwork} provides detailed training curves, illustrating convergence speed. We can tell that OFL is the closest to the original FL, converging faster than CFL and REFL. Additionally, OFL mitigates the overfitting issues in the original FL by using the opportunistic uploading strategy.
\begin{table}[ht]
\centering
\caption{Model accuracy (\%) of OFL on different datasets.}
\label{tab:utility}
{\scriptsize
\begin{tabular}{ccccc}
\toprule
Solution & original FL \cite{mcmahan2017communication} & CFL \cite{sattler2020clustered} & REFL \cite{wang2021resource} & \textbf{OFL}\\
\midrule
Cora & 83.41 $\pm$ 0.03 & 81.31 $\pm$ 0.43 & 76.74 $\pm$ 1.55 & \textbf{85.91 $\pm$ 0.11} \\
Citeseer & 71.20 $\pm$ 0.07 & 69.26 $\pm$ 0.22 & 72.45 $\pm$ 1.88 & \textbf{76.90 $\pm$ 0.09} \\
PubMed & 80.20 $\pm$ 0.01 & 78.55 $\pm$ 0.15 & 75.35 $\pm$ 2.66 & \textbf{85.23 $\pm$ 0.07} \\
MNIST & \textbf{97.70 $\pm$ 0.00} & 91.77 $\pm$ 3.13 & 91.17 $\pm$ 2.34 & 94.17 $\pm$ 0.19 \\
FMNIST & \textbf{84.28 $\pm$ 0.27} & 83.15 $\pm$ 2.93 & 73.67 $\pm$ 2.23 & 81.01 $\pm$ 0.08 \\
SVHN & \textbf{76.39 $\pm$ 0.70} & 73.65 $\pm$ 2.43 & 67.09 $\pm$ 2.18 & 73.60 $\pm$ 0.51 \\
\bottomrule
\end{tabular}
}%
\end{table}
\begin{figure}[ht]
\centering
\includegraphics[width=0.9\linewidth]{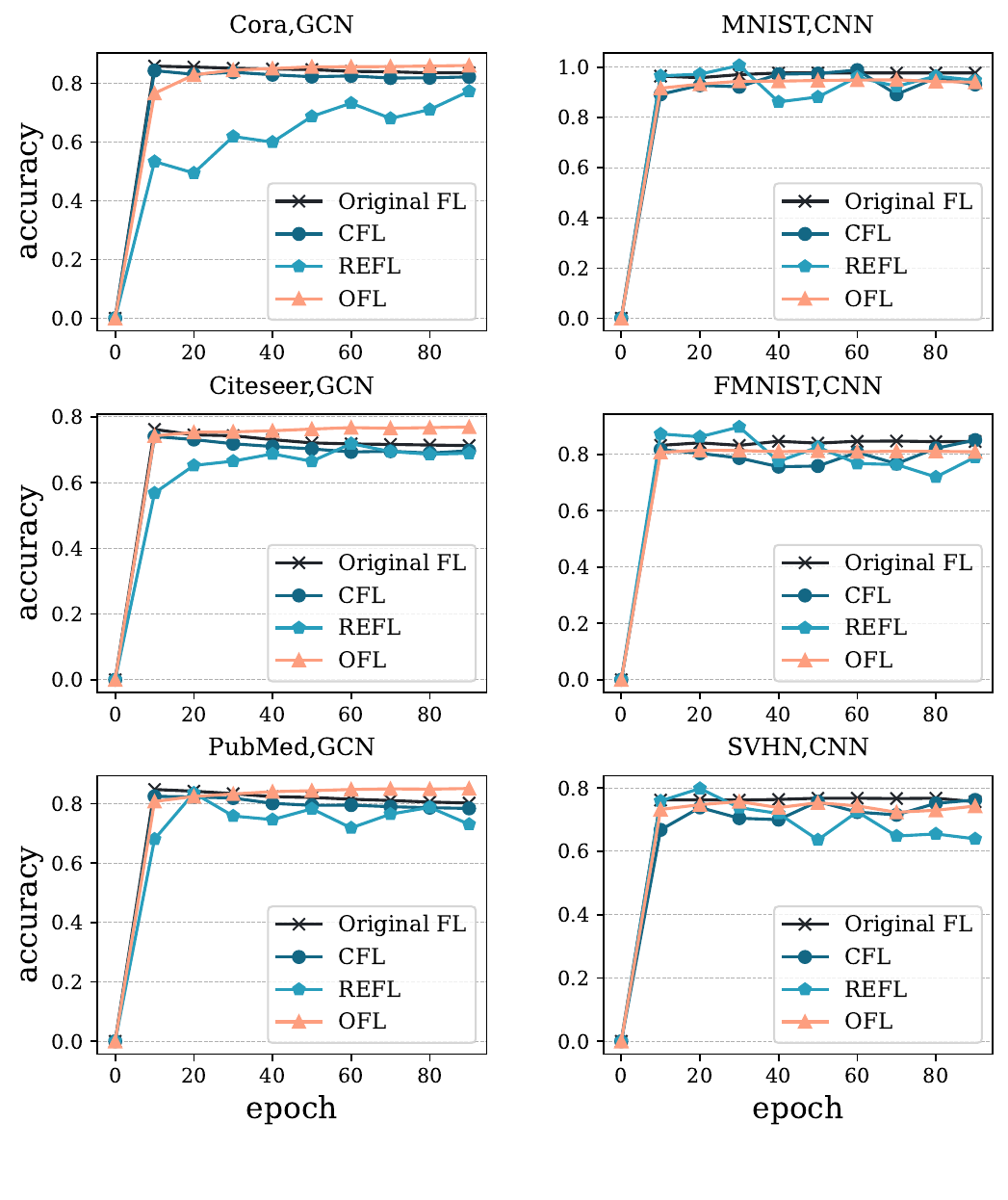}
\caption{Comparison of learning curves.}
\label{fig:ofl-relwork}
\end{figure}

We evaluate the training performance regarding different client scales. On the real-world testbed, we use $N=\{8,12,16\}$. In each case, heterogeneous devices of type A, B, C, and D have equal numbers, i.e., $\{2,3,4\}$. The results are presented in Figure~\ref{fig:ofl-n-comp-jetson}. The training processes of different client scales converge similarly, indicating that OFL's training performance is stable in small-scale tasks. To further investigate OFL's scalability, we conduct a simulation-based evaluation for large-scale FL on a high-performance server. The revaluation results of $N=\{100,300,500,800,1000\}$ are given in Figure~\ref{fig:ofl-n-comp-sim}. It is interesting to see that OFL's convergence speed slows down in large-scale graph classification tasks while OFL's convergence state slightly falls down in large-scale image classification tasks. The underlying reason is that the contrastive learning process yields more diverse subgraphs for large-scale clients, enlarging the difficulty of global consensus. As for image datasets, non-i.i.d. data split significantly affects the global model, which can be mitigated by personalized FL \cite{collins2021exploiting}.
\begin{figure}[ht]
\centering
\includegraphics[width=0.9\linewidth]{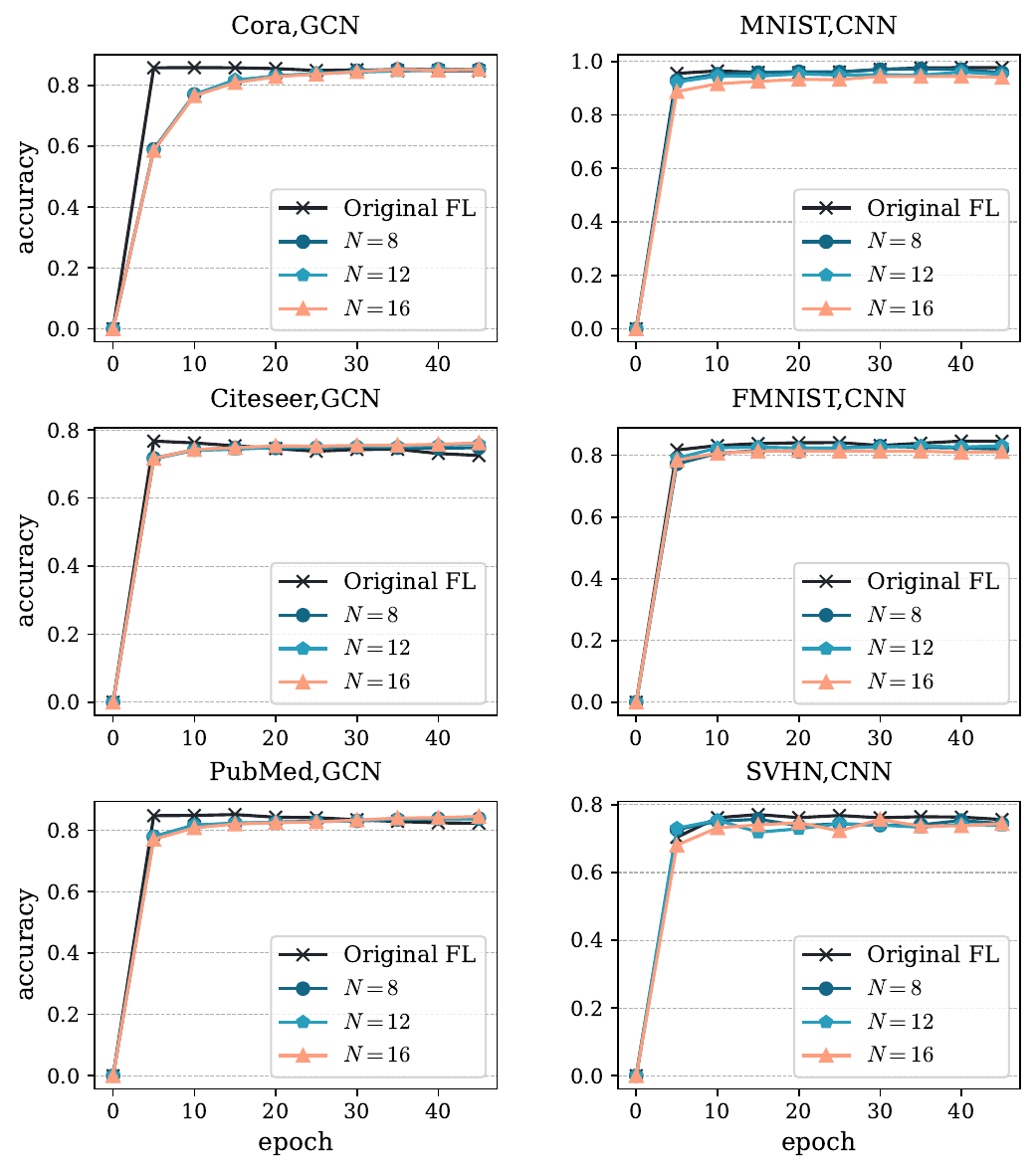}
\caption{Real-world results of different client scales.}
\label{fig:ofl-n-comp-jetson}
\end{figure}

\begin{figure}[ht]
\centering
\includegraphics[width=0.9\linewidth]{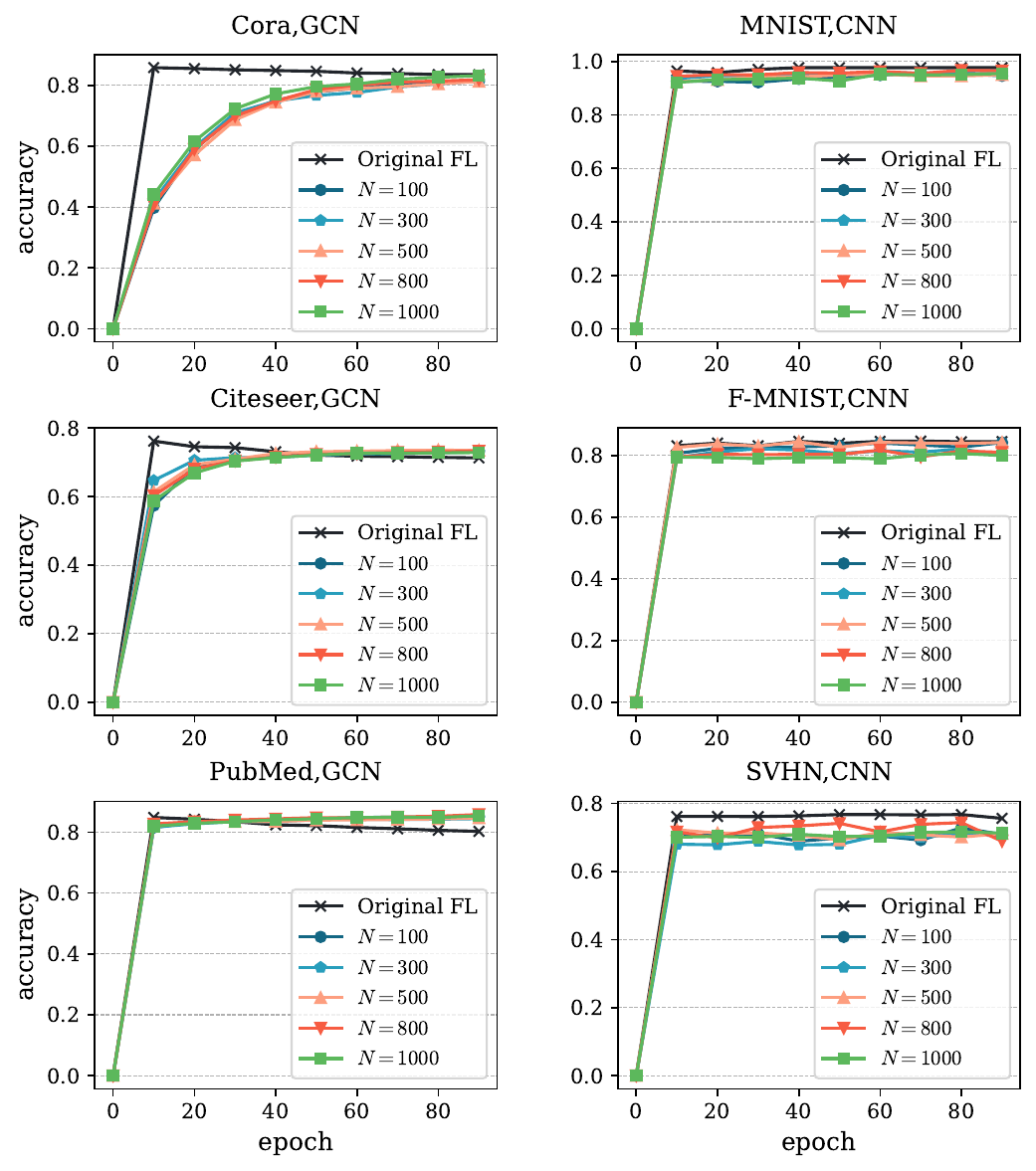}
\caption{Simulation results of large-scale tasks.}
\label{fig:ofl-n-comp-sim}
\end{figure}

We evaluate OFL's heterogeneous resources' effectiveness. In the real-world testbed, we record the training performance and cluster leader serving time of each client in Figure~\ref{fig:distribute-performance}. The left plot displays the performance variances of clients in different scales, ranging from 1\% to 11\%. The right plot illustrates the number of epochs each client serves as a leader when four devices are resource-heterogeneous. The result indicates that a device with a higher resource level is more likely to be selected as a leader. However, the leader's local model will be a straggler after long-time service. Then, devices with lower resource levels still get a chance to be leaders. In this way, OFL balances heterogeneous resources' effectiveness and training performance, offering fair model contributions.
\begin{figure}[ht]
\centering
\includegraphics[width=0.9\linewidth]{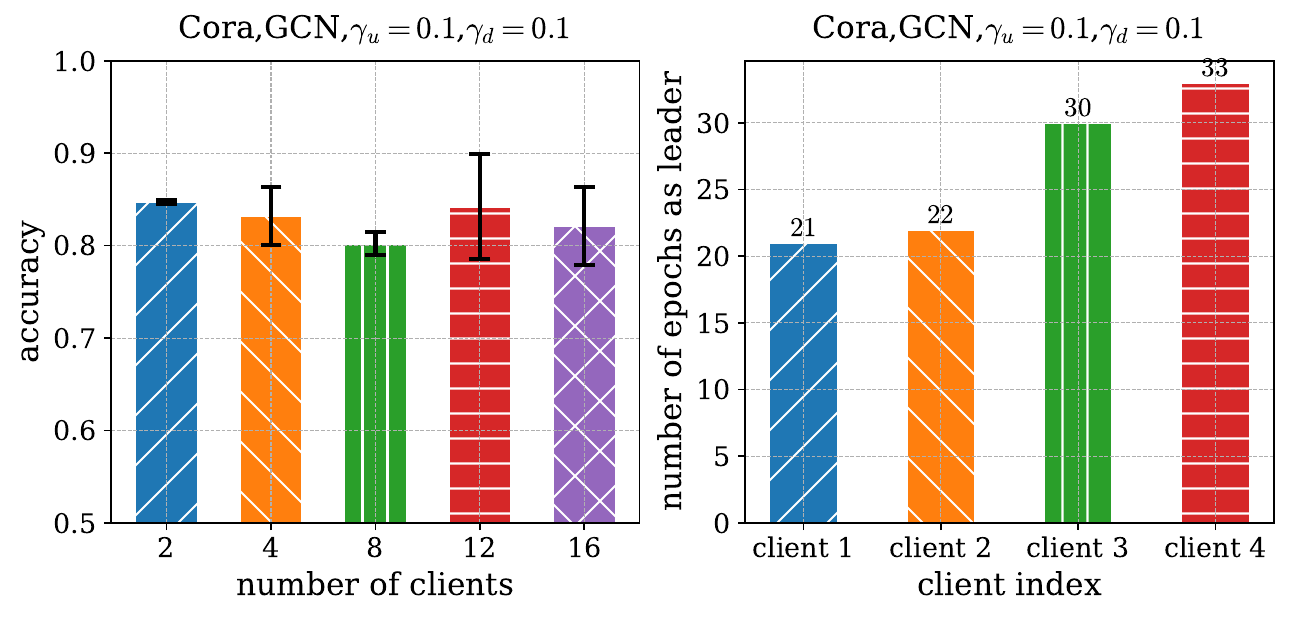}
\caption{Model accuracy and leader serving time of each client.}
\label{fig:distribute-performance}
\end{figure}

We further evaluate the impact of hyperparameters in OFL. In the evaluation, we find that the uploading and downloading ratios $\gamma_{up}$ and $\gamma_{down}$ have the most significant impact on training performance. Figures~\ref{fig:ofl-up-comp16} and \ref{fig:ofl-down-comp16} give an interesting finding that larger $\gamma_{up}$ and $\gamma_{down}$ values are not always better, which differs from the intuition. As $\gamma_{up}$ and $\gamma_{down}$ get larger, more parameters should be handled by clients and leaders, which in turn increases their workload, slowing down training performance. We also note that learning curses decrease in the last epochs because the global model becomes overfitting, which happens earlier in the original FL. This means that OFL mitigates the overfitting phenomenon by using an opportunistic syncing strategy.
\begin{figure}[ht]
\centering
\includegraphics[width=\linewidth]{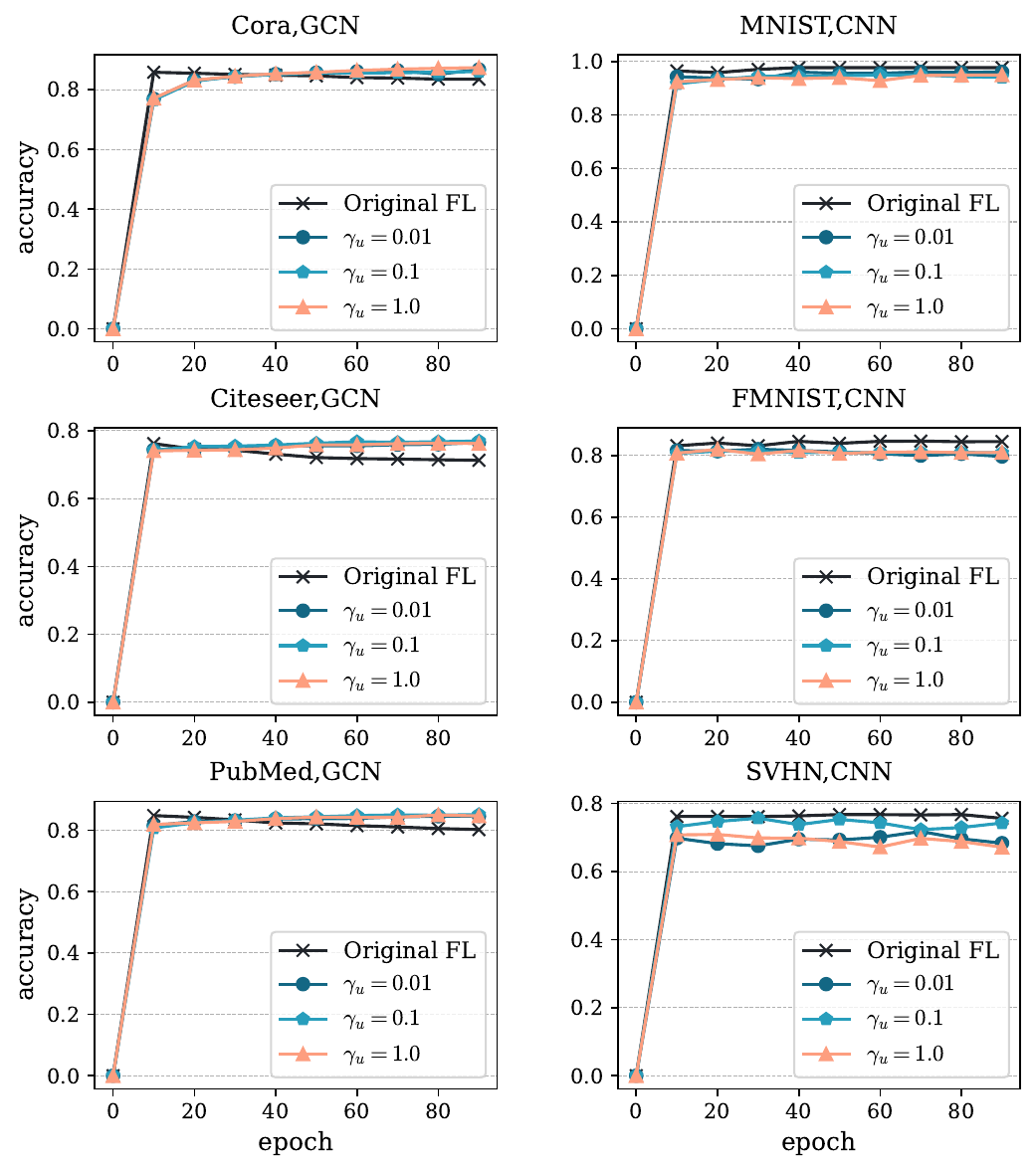}
\caption{Impact of $\gamma_{up}$ when $N=16, \gamma_{down} = 0.1$.}
\label{fig:ofl-up-comp16}
\end{figure}
\begin{figure}[ht]
\centering
\includegraphics[width=\linewidth]{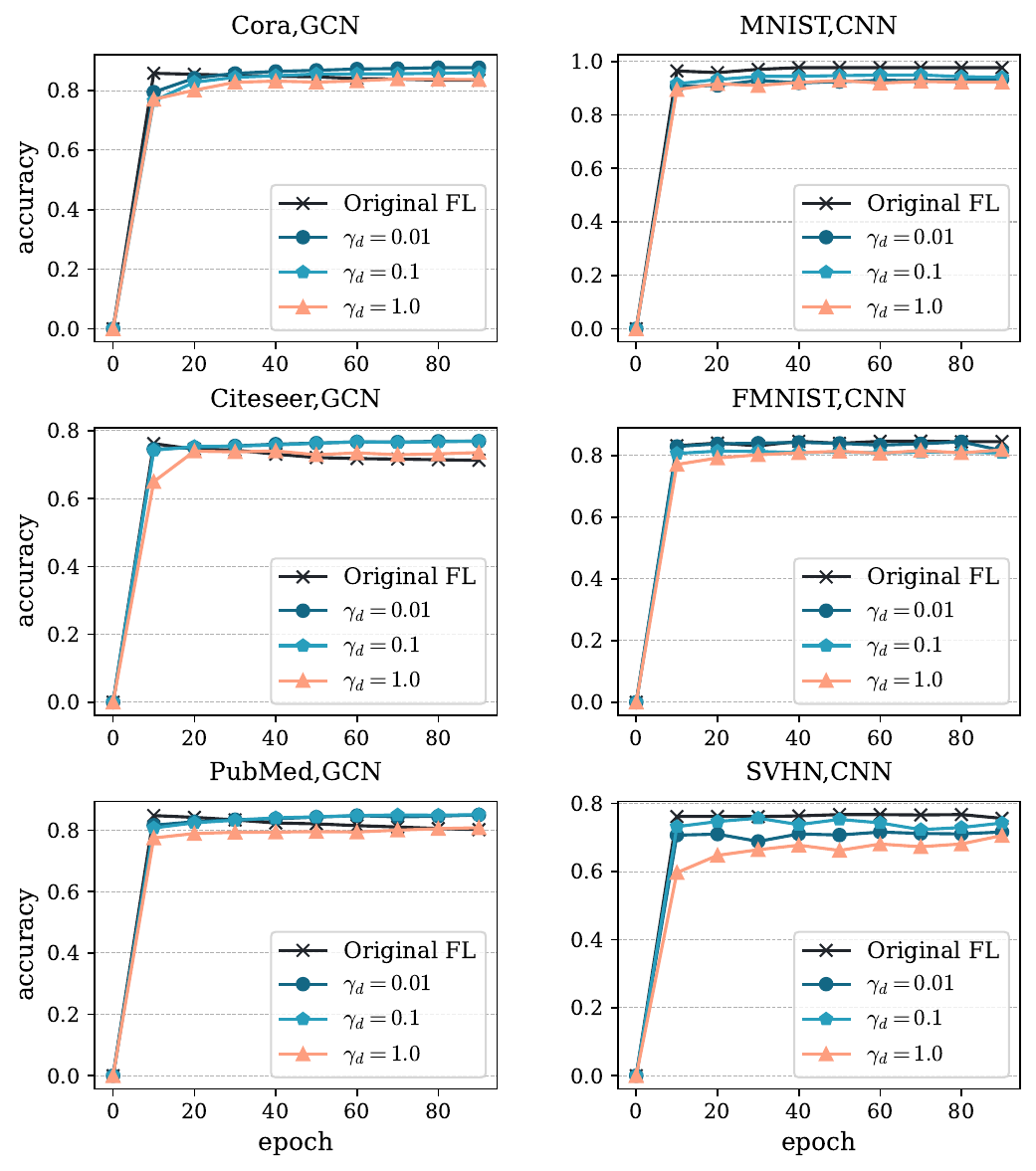}
\caption{Impact of $\gamma_{down}$ when $N=16, \gamma_{up} = 0.1$.}
\label{fig:ofl-down-comp16}
\end{figure}






\subsubsection{Efficiency}
We evaluate the efficiency of OFL from two perspectives, resource consumption and running time. To fully investigate the resource consumption of heterogeneous devices, we evaluate four different types of devices (A, B, C, and D as shown in Table~\ref{tab:setting}). In Figure~\ref{fig:workload}, we report the average memory usage, CPU usage, GPU usage, and power consumption of OFL during the training period. We can find some similar patterns in Figure~\ref{fig:workload}. Since model operators and implementation codes are very similar, heterogeneous devices share similar memory usage in all FL tasks, while figures of graph datasets have a larger variance than that of image datasets. CPU usage figures indicate that devices with more CPU cores and higher operational powers spend less CPU time, while low-resource devices need more CPU time to finish equivalent workloads. There are also some differences in Figure~\ref{fig:workload}. GPU usage is highly relevant to task types. While graph tasks are CPU-intensive, image tasks get more acceleration from GPUs. Particularly, type A devices have the least GPU usage because the bottleneck is CPU cores, causing more GPU idle time. If we zoom in on MNIST and FMNIST tasks, we can see that type B devices have more GPU usage than type C devices. Type B devices' lower maximal GPU frequency leads to more GPU time consumption. Since the SVHN dataset has more samples than the others, device GPUs of types A, B, and C are fully loaded, resulting in a clear rank regarding device resources. Since type D devices have the most powerful resources, GPU still has idle time in the SVHN case. We can see from power consumption figures that image tasks are heavier than graph tasks. Besides, power consumption figures have different results. If we ignore CPU usage, we can simply tell that power consumption has the same pattern as GPU usage, which indicates that GPU consumes more power than CPU in all FL tasks. In summary, the similarity of CPU usage shows that the CPU can be a bottleneck for heterogeneous devices, while the difference in GPU usage shows that GPU consumption significantly depends on FL tasks. Based on resource consumption results, we can conclude that grouping devices with similar resources yields a relatively stable training pace, thereby giving OFL high resource efficiency.
\begin{figure*}[ht]
\centering
\includegraphics[width=0.8\textwidth]{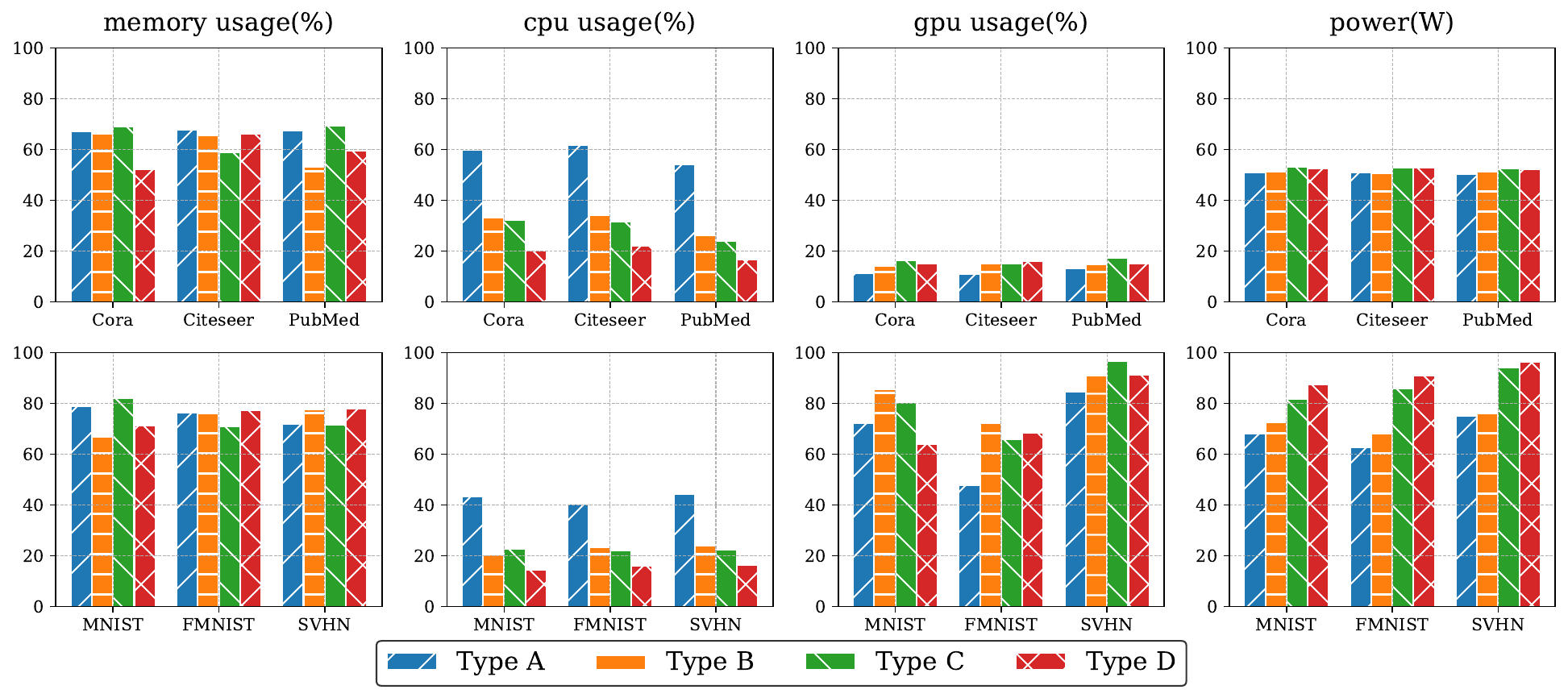}
\caption{Resource consumption of heterogeneous devices in different FL tasks.}
\label{fig:workload}
\end{figure*}


In terms of running time, we evaluate the additional time consumption of clients, cluster leaders, and the server in each task. Figure~\ref{fig:dp} shows the additional running time of each client, which is mainly caused by DP mechanisms. The time variance across different tasks is due to the data sample amount. Since DP mechanisms handle data samples locally and separately, it is reasonable to see clients' running time has a very similar pattern with dataset scales. Figure~\ref{fig:he} shows the time consumption of server-side and leader-side operations in different FL client scales. The left y-axis displays the running time of the server's additional operations (including \textsf{KeyGen}, \textsf{Add}, \textsf{ServerDec}, and \textsf{FinDec}), while the right y-axis displays the running time of each leader's additional operations (including \textsf{Enc} and \textsf{PartDec}). Notably, except for \textsf{ServerDec}, the server's time consumption for each operation grows linearly with the client scale. Benefiting from OFL's encryption design, each leader's time consumption for each operation is nearly constant.
\begin{figure}
    \centering
    \begin{minipage}{0.48\linewidth}
        \includegraphics[width=\linewidth]{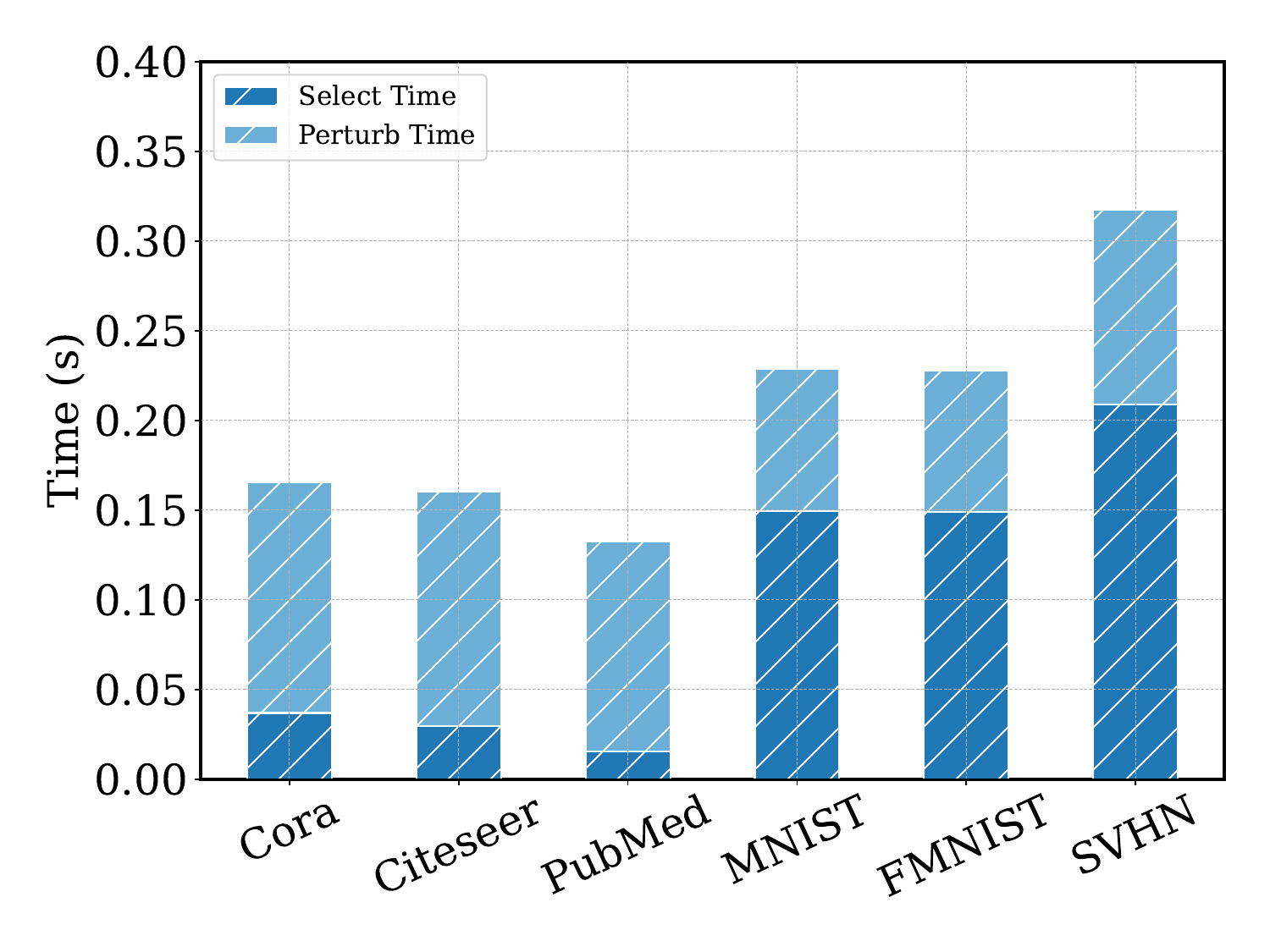}
        \caption{Running time of each client.}
        \label{fig:dp}
    \end{minipage}
    \hfill
    \begin{minipage}{0.48\linewidth}
        \includegraphics[width=\linewidth]{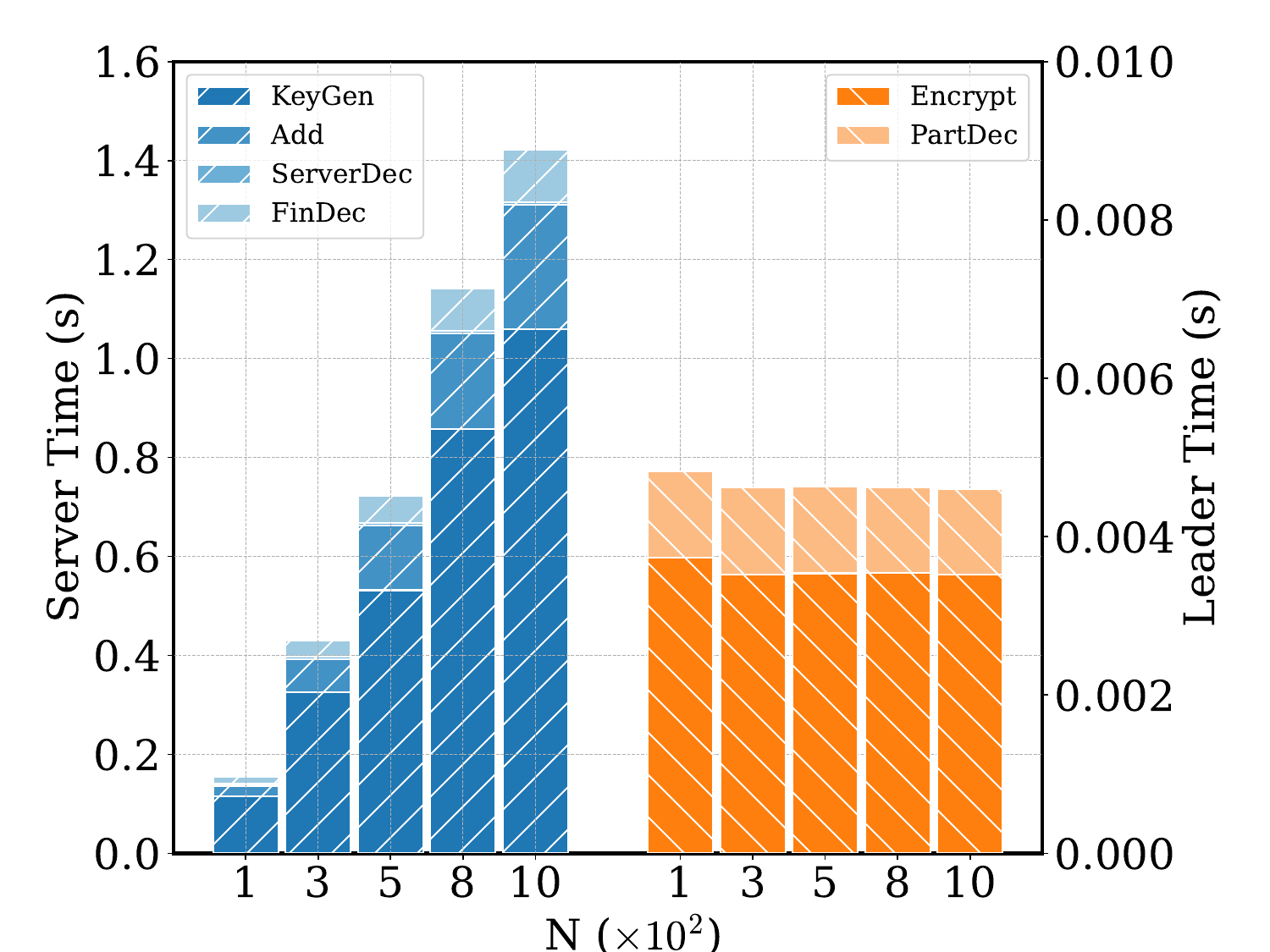}
        \caption{Running time of server and each leader.}
        \label{fig:he}
    \end{minipage}
\end{figure}

\section{Conclusion}
In this paper, we propose OFL, a novel FL framework for resource-heterogeneous and privacy-aware FL participants. OFL introduces an opportunistic synchronizing strategy to effectively reduce communication overhead and mitigate the overfitting phenomenon. By dynamic re-clustering, OFL achieves better resource utilization. OFL designs new intra-cluster DP mechanisms and inter-cluster HE schemes for security concerns. Moreover, the newly designed solution can detect outliers on encrypted model updates efficiently. Experimental results on a real-world testbed and in a large-scale simulation demonstrate that OFL achieves satisfying performance and efficiency, outperforming the existing work.

We should also note OFL's limitation. Like other resource-perception FL solutions, OFL depends on participants' resource status reports. The current OFL implementation cannot work correctly if participants cannot send or arbitrarily forge status reports. In this case, attackers can manipulate re-clustering results by reporting fake resource information, making the overall model performance more dependent on specific client parameters. This creates more favorable conditions for backdoor or poisoning attacks and is interesting to study further. Our future work will focus on an FL solution with implicit participant status observation rather than relying on explicit reports.

\bibliographystyle{ACM}
\bibliography{ref}

\end{document}